\documentclass[10pt]{article} 
\pdfoutput=1
\usepackage{amsmath,amssymb,algorithm,amsthm,,cite,subfigure,caption,cases,bm,float,subfigure,graphicx,url,color,comment}
\usepackage[hidelinks]{hyperref}
\usepackage{pbox}

\definecolor{darkred}{RGB}{150,0,0}
\definecolor{darkgreen}{RGB}{0,150,0}
\definecolor{darkblue}{RGB}{0,0,200}
\hypersetup{colorlinks=true, linkcolor=darkred, citecolor=darkgreen, urlcolor=darkblue}

\numberwithin{equation}{section}
\newtheorem{defn}{Definition}[section]
\newtheorem{prop}{Property}[section]
\newtheorem{fact}{Fact}[section]

\newtheorem{propo}{Proposition}[section]

\newtheorem{thm}{Theorem}[section]
\newtheorem{assume}{Assumption}[section]
\newtheorem{lem}{Lemma}[section]
\newtheorem{cor}{Corollary}[section]

\newcommand{\beq}{\begin{equation}}
\newcommand{\eeq}{\end{equation}}
\newcommand{\bea}{\begin{align}}
\newcommand{\eea}{{\end{align}}}

\newcommand{\U}{\mathbf{U}}
\newcommand{\V}{\mathbf{V}}

\newcommand{\Sb}{\mathbf{S}}
\newcommand{\X}{\mathbf{X}}
\newcommand{\A}{\mathbf{A}}
\newcommand{\Lb}{\mathbf{L}}

\newcommand{\prox}{\text{prox}}

\newcommand{\Dx}{{\mathbf{D}}(\x_0,}

\newcommand{\Dlf}{{\mathbf{D}}(\la\paf)}
\newcommand{\Dctf}{{\mathbf{D}}(\la\paf+T_\Cc(\x_0)^*)}
\newcommand{\Dcf}{{\mathbf{D}}(T_\Cc(\x_0)^*)}

\newcommand{\DC}{{\mathbf{D}}(\Cc)}
\newcommand{\DCC}{{\mathbf{D}}(\text{cone}(\Cc))}

\newcommand{\Delxf}{{\mathbf{D}}(\text{cone}(\paf))}
\newcommand{\Deltf}{{\mathbf{D}}(T_f(\x_0)^*)}
\newcommand{\Delxc}{{\mathbf{D}}(T_\Cc(\x_0)^*)}

\newcommand{\bu}{\text{Proj}}

\newcommand{\w}{\mathbf{w}}

\newcommand{\x}{\mathbf{x}}
\newcommand{\hax}{\hat{\mathbf{x}}}
\newcommand{\haw}{\hat{\mathbf{w}}}
\newcommand{\ub}{\mathbf{u}}
\newcommand{\g}{\mathbf{g}}
\newcommand{\vb}{\mathbf{v}}
\newcommand{\q}{\mathbf{q}}
\newcommand{\p}{\mathbf{p}}
\newcommand{\bb}{\mathbf{b}}
\newcommand{\e}{\mathbf{e}}

\newcommand{\y}{\mathbf{y}}
\newcommand{\s}{\mathbf{s}}
\newcommand{\z}{\mathbf{z}}
  
\newcommand{\ab}{\mathbf{a}}

\newcommand{\h}{\mathbf{h}}
\newcommand{\hx}{{\mathbf{x}}^{*}(\la,\z)}

\newcommand{\Sc}{{\mathcal{S}}}

\newcommand{\Kc}{\mathcal{K}}

\newcommand{\Nn}{\mathcal{N}}
\newcommand{\Cc}{\mathcal{C}}

\newcommand{\Pc}{\mathcal{P}}

\newcommand{\R}{\mathbb{R}}
\newcommand{\Pro}{{\mathbb{P}}}
\newcommand{\E}{{\mathbb{E}}}

\newcommand{\paf}{\pa f(\x_0)}
\newcommand{\pafw}{\pa f(\x_0,\w)}
\newcommand{\pafb}{\pa f(\x_0,}

\newcommand{\paw}{\pa f'({\x_0},\w)}
\newcommand{\fp}{f'({\x_0},}
\newcommand{\ff}{\hat{f}_{\x_0}}

\newcommand{\la}{{\lambda}}

\newcommand{\eps}{\epsilon}

\newcommand{\st}{\star}
\newcommand{\pa}{\partial}

\newcommand{\cl}{\text{Cl}}

\newcommand{\cn}{\text{cone}}

\newcommand{\dt}{\text{{dist}}}
\newcommand{\vs}{\vspace}
\newcommand{\hs}{\hspace}

\newcommand{\nn}{\nonumber}
\newcommand{\li}{\left<}
\newcommand{\ri}{\right>}

\newcommand{\vp}{\vspace{4pt}}
\newcommand{\Iden}{{\bf{I}}}

\newcommand{\order}[1]{\mathcal{O}\left(#1\right)}
\newcommand{\Keywords}[1]{\par\noindent
{\small{\em \indent{\bf{Keywords}}\/}: #1}}

\begin{document}
\title{\Large{\bf{Sharp MSE Bounds for Proximal Denoising}}\vs{25pt}}
\author{Samet Oymak\hs{120pt}Babak Hassibi\thanks{ 
This work was supported in part by the National Science Foundation under grants CCF-0729203, CNS-0932428 and CIF-1018927, by the Office of Naval Research under the MURI grant N00014-08-1-0747, and by a grant from Qualcomm Inc.}\vs{15pt}\\Department of Electrical Engineering \\ Caltech, Pasadena - 91125\\\normalsize \href{mailto:soymak@caltech.edu}{soymak@caltech.edu}, \href{mailto:hassibi@caltech.edu}{hassibi@caltech.edu}} 
\date{}
\maketitle

\vspace{10pt}

\begin{abstract} 
Denoising has to do with estimating a signal $\x_0$ from its noisy observations $\y=\x_0+\z$. In this paper, we focus on the ``structured denoising problem", where the signal $\x_0$ possesses a certain structure and $\z$ has independent normally distributed entries with mean zero and variance $\sigma^2$. We employ a structure-inducing convex function $f(\cdot)$ and solve $\min_\x\{\frac{1}{2}\|\y-\x\|_2^2+\sigma\la f(\x)\}$ to estimate $\x_0$, for some $\lambda>0$. Common choices for $f(\cdot)$ include the $\ell_1$ norm for
sparse vectors, the $\ell_1-\ell_2$ norm for block-sparse signals and the nuclear norm for low rank
matrices. The metric we use to evaluate the performance of an estimate $\x^*$ is the normalized mean-squared-error $\text{NMSE}(\sigma)=\frac{\E\|\x^*-\x_0\|_2^2}{\sigma^2}$. We show that NMSE is maximized as $\sigma\rightarrow 0$ and we find the \emph{exact} worst case NMSE, which has a simple geometric interpretation: the mean-squared-distance of a standard normal vector to the $\la$-scaled subdifferential $\la\paf$. When $\la$ is optimally tuned to minimize the worst-case NMSE, our results can be related to the constrained denoising problem $\min_{f(\x)\leq f(\x_0)}\{\|\y-\x\|_2\}$. The paper also connects these results to the generalized LASSO problem, in which, one solves $\min_{f(\x)\leq f(\x_0)}\{\|\y-\A\x\|_2\}$ to estimate $\x_0$ from noisy linear observations $\y=\A\x_0+\z$. We show that certain properties of the LASSO problem are closely related to the denoising problem. In particular, we characterize the normalized LASSO cost and show that it exhibits a ``phase transition'' as a function of number of observations. Our results are significant in two ways. First, we find a simple formula for the performance of a general convex estimator. Secondly, we establish a connection between the denoising and linear inverse problems.

\vspace{5pt}
\end{abstract}



\section{Introduction}

%



Signals exhibiting structured behavior play a critical role in
various applications. In particular, sparse signals, block-sparse signals, low-rank matrices and their many variations are often the underlying solutions of problems, with applications ranging from MRI to recommendation systems to DNA microarrays, \cite{CandesMain,Cai,Cha,Bach,Don1,soft-thresh,Comp,MatCom,Block}. Hence, a significant amount of work has been dedicated to developing and analyzing algorithms, that can take advantage of the signal structure. In this work, we will be
considering the estimation of structured signals corrupted by
additive noise via convex optimization. Under Gaussian noise assumption, we will provide an exact characterization of the estimation performances of useful convex algorithms based on \emph{proximity operator}. Proximity operator corresponding to a function $f(\cdot)$ at a point $\y$ is given by,
\beq
\prox_f(\y,\la)=\arg\min_\x\frac{1}{2}\|\y-\x\|_2^2+\la f(\x)\label{proxmain}
\eeq
Proximity operator was first introduced by Moreau in \cite{More}. It has several \emph{nice} properties, it can often be evaluated quickly and it constitutes the backbone of the proximal algorithms, \cite{More,prox2,prox3,prox4,Hale,Combet,Combet2,Boyd2}. The topic of this work is to understand and characterize the estimation capabilities of $\prox_f(\cdot)$ when $\y$ is the noisy observations of an underlying structured signal $\x_0$. Our results will be particularly meaningful when $f(\cdot)$ is \emph{structure inducing}. The prime example is the $\ell_1$ norm, in which case \eqref{proxmain} is known as ``soft-thresholding the entries of $\y$'' and has the following closed form solution, \cite{soft-thresh,Donoho,AMPmain,Tikh,Combet},
\begin{align}
[\prox_f(\y,\la)]_i=\begin{cases}y_i-\la~\text{if}~y_i\geq \la\\0~\text{if}~|y_i|< \la\\y_i+\la~\text{else}\end{cases}.
\end{align}
While sparse signal estimation is the fundamental question in compressed sensing and properties of $\ell_1$ minimization has been studied extensively, recent literature shows that, various other signal forms show up in diverse set of applications and these applications are growing in a daily basis, \cite{Koltc,Plan,Oym,Fazel,Block,RechtFazel,RechtBlock,Model,Nonuniform,Oym2,Vasvani,LPS,PCA,NoisyLPS,Negah,MinTao,Emile3}. Consequently, a uniform treatment of the structured signal recovery problems is highly desirable. With this intention, we will loosely say, $\x_0$ is a structured signal and $f(\cdot):\R^n\rightarrow \R$ is a convex function that exploits this particular structure. Such structured signal-function pairs include the sparse vectors and the $\ell_1$ norm, and the low-rank matrices and the nuclear norm. Chandrasekaran et al. \cite{Cha} and Bach \cite{Bach} consider systematic ways of finding the structure-inducing function given the characteristics of the signal. 

With this motivation, we will provide sharp estimation guarantees for the general problem \eqref{proxmain} where $\y=\x_0+\sigma\vb$, $\x_0$ is a structured signal, and $\vb$ is the noise vector; whose entries are independent standard normal. As it has been observed in the relevant literature, \cite{Rec2,Mon,BayMon,OymLAS}, when noise has variance $\sigma^2$, it is useful to consider the slight modification of \eqref{proxmain},
\beq
\prox_f(\y,\sigma\la)=\arg\min_\x\frac{1}{2}\|\y-\x\|_2^2+\sigma\la f(\x)\label{proxmain2},
\eeq
which takes the normalization for $\sigma$ into account. We focus on finding a tight upper bound on the normalized-mean-squared-error (NMSE) defined as,
\beq
\frac{\E[\|\prox_f(\x_0+\sigma\vb,\sigma\la)-\x_0\|^2_2]}{\sigma^2}\label{NMSE}.
\eeq
NMSE \eqref{NMSE} is a function of variance $\sigma^2$, vector $\x_0$ and function $f(\cdot)$. We find a formula for the highest value of NMSE over $\sigma>0$, which is only a function of $f(\cdot)$ and $\x_0$. To state our result, we need to introduce some notation.
\begin{itemize}
\item \emph{Subdifferential:} The set of subgradients of a convex function $f(\cdot)$ at $\x_0$ will be denoted by $\paf$.
\item \emph{Distance:} Given a nonempty set $\Cc$, the distance of a vector $\vb$ to $\Cc$ is $\dt(\vb,\Cc):=\inf_{\ub\in\Cc}\|\vb-\ub\|_2$.
\item \emph{Mean-squared-distance (MSD):} Let $\g\in\R^n$ be a vector with standard normal entries and $\Cc$ be a nonempty set. Define $\DC=\E[\dt(\g,\Cc)^2]$.
\end{itemize}

The following theorem gives a sample result.

\begin{thm}[Worst Case NMSE] \label{proximal} Assume $f(\cdot):\R^n\rightarrow\R$ is a convex function, $\x_0\in\R^n$, $\la\geq 0$ and $\vb$ has independent standard normal entries. Then,
\beq
\max_{\sigma>0}\frac{\E[\|\prox_f(\x_0+\sigma\vb,\sigma\la)-\x_0\|^2_2]}{\sigma^2}=\Dlf\label{regprox}.
\eeq
Furthermore, the worst case NMSE is achieved as $\sigma\rightarrow 0$.
\end{thm}

\noindent{\bf{Remark 1:}} Observe that, the result is not interesting if the function is differentiable at $\x_0$, in which case, subdifferential is a singleton. It becomes useful when the subdifferential is large; which decreases the distance term $\Dlf$.

\noindent{\bf{Remark 2:}} The fact that worst case NMSE is achieved as $\sigma\rightarrow 0$ has been observed in relevant problems which will be discussed later on \cite{Matesh, NoiseSense,OymLAS, Donoho}

The quantity $\Dlf$ may seem abstract at first sight. However, for the problems of interest, the subdifferential $\paf$ is a highly structured and well-studied set, \cite{DonCentSym,CandesMain,Decomp1}. Several works, \cite{Foygel,McCoy,OymLAS,Cha}, provide useful upper bounds on this quantity for certain structured signal classes.

\subsection{Examples} \label{examples}
We will now list some specific examples which are applications of Theorem \ref{proximal} in combination with Table \ref{table:closed}.
\begin{table}[h]
\begin{center}
\hspace*{-5pt}  \begin{tabular}{ | c |  c  |  c | c |}
     \cline{2-4}
     \multicolumn{1}{c|}{}&\pbox{20cm}{\vp{\bf{{{$k$-sparse, $\x_0\in\R^n$}}}}\vp}& \pbox{20cm}{\bf{{{Rank $r$, $\X_0\in\R^{d\times d}$}}}}   & $k$-{\bf{block sparse}}, $\x_0\in\R^{tb}$ \\     \hline
     {\color{darkred}{{$\Dlf$}}}  & $ (\la^2+3)k$ ~for~ $\la\geq \sqrt{2\log \frac{n}{k}}$ & \pbox{20cm}{\vp \small{$\la^2r+2d(r+1)$ ~for~ $\la\geq 2\sqrt{d}$}\vp} & \small{$(\la^2+b+2)k$~ for~ $\la \geq \sqrt{b}+\sqrt{2\log\frac{t}{k}}$} \\
    \hline
  \end{tabular}
  \caption{Closed form upper bounds for $\Dlf$ corresponding to \eqref{eq:clo1}, \eqref{eq:clo2} and \eqref{eq:clo3} (from \cite{Foygel,OymLAS}).}
  \label{table:closed}
\end{center}
\end{table}


 \noindent{$~~\emph1$}.  {\bf{Sparse signal estimation:}} Assume $\x_0\in\R^n$ has $k$ nonzero entries and $\vb$ has independent standard normal entries and $\y=\x_0+\sigma\vb$. Pick $f(\cdot)$ as the $\ell_1$ norm. Then, for $\la\geq \sqrt{2\log\frac{n}{k}}$,
\beq\frac{\E[\|\prox_{\ell_1}(\y,\sigma\la)-\x_0\|_2^2]}{\sigma^2}\leq (\la^2+3)k.\label{eq:clo1}\eeq

 \noindent{$~~\emph2$}. {\bf{Low-rank matrix estimation:}} Assume $\X_0\in\R^{d\times d}$ is a rank $r$ matrix, $n=d\times d$. This time, $\x_0\in\R^n$ corresponds to vectorization of $\X_0$ and $f(\cdot)$ is chosen as the nuclear norm $\|\cdot\|_\st$ (sum of the singular values of a matrix) \cite{Fazel,RechtFazel}. Hence, we observe $\y=\text{vec}(\X_0)+\sigma\vb$ and estimate $\X_0$ as,
\beq\prox_{\st}(\y,\sigma\la)=\arg\min_{\X}\frac{1}{2}\|\y-\text{vec}(\X)\|_2^2+\sigma\la\|\X\|_\st.\nn\eeq

Then, whenever $\la\geq 2\sqrt{d}$, normalized mean-squared-error satisfies,
\beq\frac{\E[\|\prox_{\st}(\y,\sigma\la)-\X_0\|^2_F]}{\sigma^2}\leq (\la^2+2d)r+2d.\label{eq:clo2}\eeq

 \noindent{$~~\emph3$}. {\bf{Block sparse estimation:}} Let $n=t\times b$ and assume the entries of $\x_0\in\R^n$ can be grouped into $t$ known blocks of size $b$ so that only $k$ of these $t$ blocks are nonzero. To induce the structure, the standard approach is to use the $\ell_{1,2}$ norm which sums up the $\ell_2$ norms of the blocks, \cite{Yonina,Block,RechtBlock}. In particular, denoting the subvector corresponding to $i$'th block of a vector $\x$ by $\x_i$, the $\ell_{1,2}$ norm is equal to $\|\x\|_{1,2}=\sum_{i=1}^t\|\x_i\|_2$.
Pick $f(\cdot)=\|\cdot\|_{1,2}$ and assume $\la \geq \sqrt{b}+\sqrt{2\log\frac{t}{k}}$, $\vb$ has independent standard normal entries and $\y=\x_0+\sigma\vb$.
\beq\frac{\E[\|\prox_{\ell_{1,2}}(\y,\sigma\la)-\x_0\|^2_2]}{\sigma^2}\leq (\la^2+b+2)k.\label{eq:clo3}\eeq

\subsection{Relevant literature}
Proximity operator with $\ell_1$ minimization has been subject of various works and results of type \eqref{eq:clo1} is known \cite{DonPhaseTrans,Cha, soft-thresh,BayMon}. Block sparse signals have been studied in \cite{Donoho,Yonina}. NMSE properties of low-rank matrices has been analyzed in detail by \cite{Matesh,Shabal,Matesh2}. Our main contribution is giving the \emph{exact} worst case NMSE for arbitrary convex functions. As a side result, we provide closed form upper bounds such as \eqref{eq:clo1}, \eqref{eq:clo2} and \eqref{eq:clo3}. Closer to our results, in \cite{Rec2}, Bhaskar et al. considers denoising via structure inducing ``atomic norms'' with a focus on line spectral estimation. However, their error bounds are looser than what Theorem \ref{proximal} gives. For example, in \eqref{eq:clo2}, we are able to bound the error in terms of the sparsity of the vector; whereas \cite{Rec2} bounds in terms of the $\ell_1$ norm of the vector which can be \emph{substantially} larger. In \cite{ChaJor}, Chandrasekaran and Jordan consider the ``constrained denoising'' problem which enforces the constraint $f(\x)\leq f(\x_0)$ rather than penalization $\sigma\la f(\x)$ in \eqref{proxmain2}. As it will be discussed in Section \ref{sec:main}, their results are consistent with ours; however, we additionally show that NMSE bounds are sharp as $\sigma\rightarrow 0$.

A more challenging problem occurs when $\y$ is arising from noisy \emph{linear} observations $\A\x_0$ and we have to solve an underdetermined linear inverse problem. In this case, we can consider a variant of \eqref{proxmain}, namely,
\beq
\arg\min_\x\frac{1}{2}\|\y-\A\x\|_2^2+\la f(\x)\label{proxlas}.
\eeq
When $f(\cdot)$ is the $\ell_1$ norm and $\x_0$ is a sparse signal \eqref{proxlas} is known as LASSO. LASSO has been the subject of great interest as it is a natural and powerful approach to noise robust compressed sensing \cite{Tikh, BicRit,BunTsy,StojLAS,StojSOCP,ExampleLAS,Koltc}. Section \ref{sec:main} provides certain results on LASSO and connection of our results to the linear inverse problems in more detail.


Some of the advantages of our results are as follows.
 \begin{itemize}
\item In various compressed sensing or estimation results, guarantees are orderwise rather than the exact quantities. For example, while NMSE of $\order{dr}$ is orderwise optimal for nuclear norm minimization (recall \eqref{eq:clo2}), it is often critical to know the actual constant term that multiplies $dr$ for real life applications. For Gaussian noise, it is known that, this constant is as small as $6$ rather than being, say, $1000$ (set $\la=2\sqrt{d}$ in \eqref{eq:clo2}). Consequently, a \emph{formula} that gives the exact performance is desirable.
\item Instead of finding case specific results for various structured signal classes, we are stating our results for a generic convex function $f(\cdot)$; which allows us to systematically obtain the specific instances. The only required information is the subdifferential $\paf$.
\item As it will be discussed in Section \ref{sec:main}, we will relate Theorem \ref{proximal} to the linear inverse problem \eqref{proxlas} and in fact, we will find a general relation between them, in connection to the recent results of \cite{Cha,McCoy}. In particular, for both problems, the subdifferential based quantity, $\Dlf$, plays a critical role.
\end{itemize}

To provide a better intuition about our results, we will now exemplify other possible signal forms and the associated structure inducing functions.

\section{Common structures and the associated functions}\label{commonfunc}
We consider the following low dimensional signal models and for each case, we explain the signal structure and provide the convex function $f(\cdot)$ that exploits the structure.

%
\vspace{2pt}
\noindent$\bullet$ {\bf{Sum of a low rank and a sparse matrix:}} The matrix of interest $\X_0$ can be decomposed into a low rank piece $\Lb_0$ and a sparse piece $\Sb_0$, hence it is a ``mixture'' of the low rank and sparse structures. This model is useful in applications such as video surveillance and face recognition, \cite{PCA,NoisyLPS}. In this case, function $f(\cdot)$ should be chosen so that it simultaneously emphasizes sparse and low rank pieces in the given matrix.
\beq\label{LPS}
f(\X)=\inf_{\Lb+\Sb=\X}\|\Lb\|_\star+\gamma \|\Sb\|_1
\eeq
\eqref{LPS} is known as the infimal convolution of the $\ell_1$ norm and the nuclear norm, \cite{PCA,LPS}.

\vspace{3pt}
\noindent$\bullet$ {\bf{Discrete total variation:}} In many imaging applications, \cite{Rudin,TotalCompress1,TotalCompress2}, the signal of interest $\x_0$ rarely changes as a function of time. Consequently, letting ${\bf{d}}_i=\x_{0,i+1}-\x_{0,i}$ for $1\leq i\leq n-1$, ${\bf{d}}\in\R^{n-1}$ becomes a sparse vector. To induce this structure, one may minimize the total variation of $\x_0$, namely $\|\x_0\|_{TV}=\|{\bf{d}}\|_1$.

\vspace{3pt}
\noindent$\bullet$ {\bf{Nonuniformly sparse signals and weighted $\ell_1$ minimization:}} In many cases, we might have prior information regarding the sparsity pattern of the signal, \cite{CandesReweight,Nonuniform,Oym2,Vasvani}. In particular, the signal $\x_0$ might be relatively sparse over a certain region and dense over another. To exploit this additional information we can use a modified $\ell_1$ minimization where different weights are assigned to different regions. More rigorously, assume the set of entries $\{1,2,\dots,n\}$ is divided into $t$ disjoint sets $S_1,\dots,S_t$ that correspond to regions with different sparsity levels. Then, given a nonnegative weight vector $\w=[w_1~w_2\dots~w_t]$, weighted $\ell_1$ norm $\|\x\|_w$ can be given as:
\beq
\|\x\|_w=\sum_{i=1}^tw_i\sum_{j\in S_i}|x_j|\nn
\eeq

\noindent$\bullet$ {\bf{Other models:}} We can include various other models: signals that are sparse and positive, \cite{DonNonNegative,Amin}; positive semidefinite constraints, \cite{Fazel,Oym}; simultaneously sparse and low rank matrices, \cite{simultaneous,French,Emile3}; permutation matrices, binary vectors, cut matrices, \cite{Cha}; etc.

\section{Main Contributions}\label{sec:main}
\subsection{Notation}\label{notate1}
Before stating our main results, we will introduce the relevant notation.
 $\Nn(0,\sigma^2\Iden_{n})$ is used to denote the distribution of a vector in $\R^n$ with independent Gaussian entries with variance $\sigma^2$ and mean zero. For a convex function $f(\cdot):\R^n\rightarrow \R$, the set of subgradients of $f(\cdot)$ at $\x$ is denoted by $\pa f(\x)$. $\pa f(\x)$ is a nonempty, convex and compact set, \cite{Urru}. Given a set $A$ and $\la\in\R$, $\la A$ will denote the set obtained by scaling elements of $A$ by $\la$. The cone obtained by $A$ is given as,
\beq\nn
\cn(A)=\{\la\x\in\R^n\big|\x\in A,\la\geq 0\}.
\eeq
When $A,B$ are two sets, $A+B$ will denote the Minkowski sum $\{\ab+\bb\big|\ab\in A~\text{and}~\bb\in B\}$. The closure of a set $A$ will be denoted by $\cl(A)$. We will now briefly describe our notation on convex geometry. The distance between a set nonempty set $\Cc$ and a point $\x\in\R^n$ is given as: $\dt(\x,\Cc)=\min_{\s\in\Cc}\|\x-\s\|_2$. When $\Cc$ is closed and convex, there exists a unique point in $\Cc$ that is closest to $\x$ called the ``projection of $\x$ on $\Cc$''. This point will be denoted by $\bu(\x,\Cc)$. Basically, $\bu(\x,\Cc)\in\Cc$ and,
\beq
\|\x-\bu(\x,\Cc)\|_2=\dt(\x,\Cc).\nn
\eeq
The polar cone of $\Cc$, which is always closed and convex, will be denoted by $\Cc^*$ and is given as:
\beq
\Cc^*=\{\vb\big|\li\vb,\x\ri\leq 0~\text{for all}~\x\in\Cc\}.\nn
\eeq

\subsubsection{Useful concepts}
\vspace{3pt}
\noindent{\bf{Mean-squared-distance}} Let $\Cc$ be a nonempty subset of $\R^n$. The mean-squared-distance (MSD) to $\Cc$ will be denoted by $\DC$ and is defined as,
\beq
\DC=\E[\dt(\g,\Cc)^2],
\eeq where $\g\sim\Nn(0,\Iden)$. The relevant definitions have been used in \cite{Foygel,OymLAS}. This definition is closely related to the concepts like ``statistical dimension'' of \cite{McCoy} and ``Gaussian width'' of \cite{Cha}.

\vspace{2pt}

\noindent{\bf{Feasible sets and tangent cones:}} Given a convex function $f(\cdot):\R^n\rightarrow \R$, the descent set at $\x_0$ is given as $F_f(\x_0)=\{\z\big|f(\x_0+\z)\leq f(\x_0)\}$. Then, the tangent cone at $\x_0$ is defined as $T_f(\x_0)=\cl(\cn(F_f(\x_0)))$. In particular, when $\x_0$ is not a minimizer of $f(\cdot)$, we have $T_f(\x_0)=\cn(\pa f(\x_0))^*$. In fact, this will be the only assumption we will be making besides the convexity of $f(\cdot)$. The details of this result can be found in \cite{Bertse,Roc70}. In a similar manner, for a convex set $\Cc$ and $\x_0\in\Cc$, the set of feasible directions at $\x_0$ is denoted by $F_\Cc(\x_0)$ and is given by $F_\Cc(\x_0)=\{\z\big|\x_0+\z\in\Cc\}$. Then, the tangent cone at $\x_0$ becomes $T_\Cc(\x_0):=\cl(\cn(F_\Cc(\x_0)))$.

\subsection{Results on denoising}
Throughout the discussion below, we assume $\x_0$ is \emph{not} a minimizer of $f(\cdot)$.
\begin{thm}[Constrained denoising] \label{constrained} Let $\Cc$ be a nonempty, convex and closed set and let $\x_0$ be an arbitrary vector in $\Cc$. Let,
\beq
\prox_\Cc(\y)=\arg\min_{\x\in\Cc}\|\y-\x\|_2\label{constrained prox}.
\eeq
Assume $\vb\sim\Nn(0,\Iden_n)$. Then,
\beq
\max_{\sigma>0}\frac{\E[\|\prox_\Cc(\x_0+\sigma\vb)-\x_0\|_2^2]}{\sigma^2}=\Dcf\label{constrained eq}.
\eeq
Furthermore, the equality above is achieved as $\sigma\rightarrow 0$. Now, let $f(\cdot):\R^n\rightarrow \R$ be a convex function and choose $\Cc=\{\x\in\R^n\big| f(\x)\leq f(\x_0)\}$. For this choice, we have,
\beq
\max_{\sigma>0}\frac{\E[\|\prox_\Cc(\x_0+\sigma\vb)-\x_0\|_2^2]}{\sigma^2}=\Deltf=\Delxf\label{constprox}.
\eeq

\end{thm}

\noindent{\bf{Remark:}} The constrained and the regularized denoising problems are closely related. To relate these, consider the right hand sides of \eqref{constprox} and \eqref{regprox}, in particular, the quantities, $\Delxf$ and $\Dlf$. Since $\la\paf$ is subset of $\text{cone}(\paf)$, for all $\la\geq 0$, one has,
\beq
\Delxf\leq\Dlf\nn.
\eeq
On the other hand, as it will be discussed later on, when $f(\cdot)$ is a norm, making use of the results of \cite{McCoy}, we can actually obtain,
\beq
0\leq  \min_{\la\geq0}\Dlf-\Delxf\leq 2\frac{\sup_{\s\in\paf}\|\s\|_2}{f(\frac{\x_0}{\|\x_0\|_2})}\label{sandwich}.
\eeq

\eqref{sandwich} can be simplified for structure inducing functions. For example, when $\x_0$ is a $k$ sparse vector in $\R^n$, right hand side can be replaced by $2\sqrt{\frac{n}{k}}$. For the low-rank and block-sparse signals described in Section \ref{examples}, we can use $2\sqrt{\frac{d}{r}}$ and $2\sqrt{\frac{t}{k}}$ respectively.

Our next result considers a mixture of constrained and regularized estimators.


\begin{thm}[General upper bound] \label{con mini} Assume $\Cc$ is a nonempty, convex and closed set and $f(\cdot):\R^n\rightarrow\R$ is a convex function. Assume $\x_0\in\Cc$, $\la\geq 0$ and $\vb\sim\Nn(0,\Iden_n)$. Consider the following estimator,
\beq\label{prox}
\prox_{f,\Cc}(\y,\sigma\la)=\arg\min_{\x\in\Cc} \sigma\la f(\x)+\frac{1}{2}\|\y-\x\|_2^2.
\eeq
For any $\sigma>0$, we have,
\beq
\frac{\E[\|\prox_{f,\Cc}(\x_0+\sigma\vb,\sigma\la)-\x_0\|_2^2]}{\sigma^2}\leq\Dctf\label{conmini}.
\eeq
\end{thm}

 \begin{figure}

  \begin{center}
{\includegraphics[width=0.6\textwidth,height=0.23\textheight]{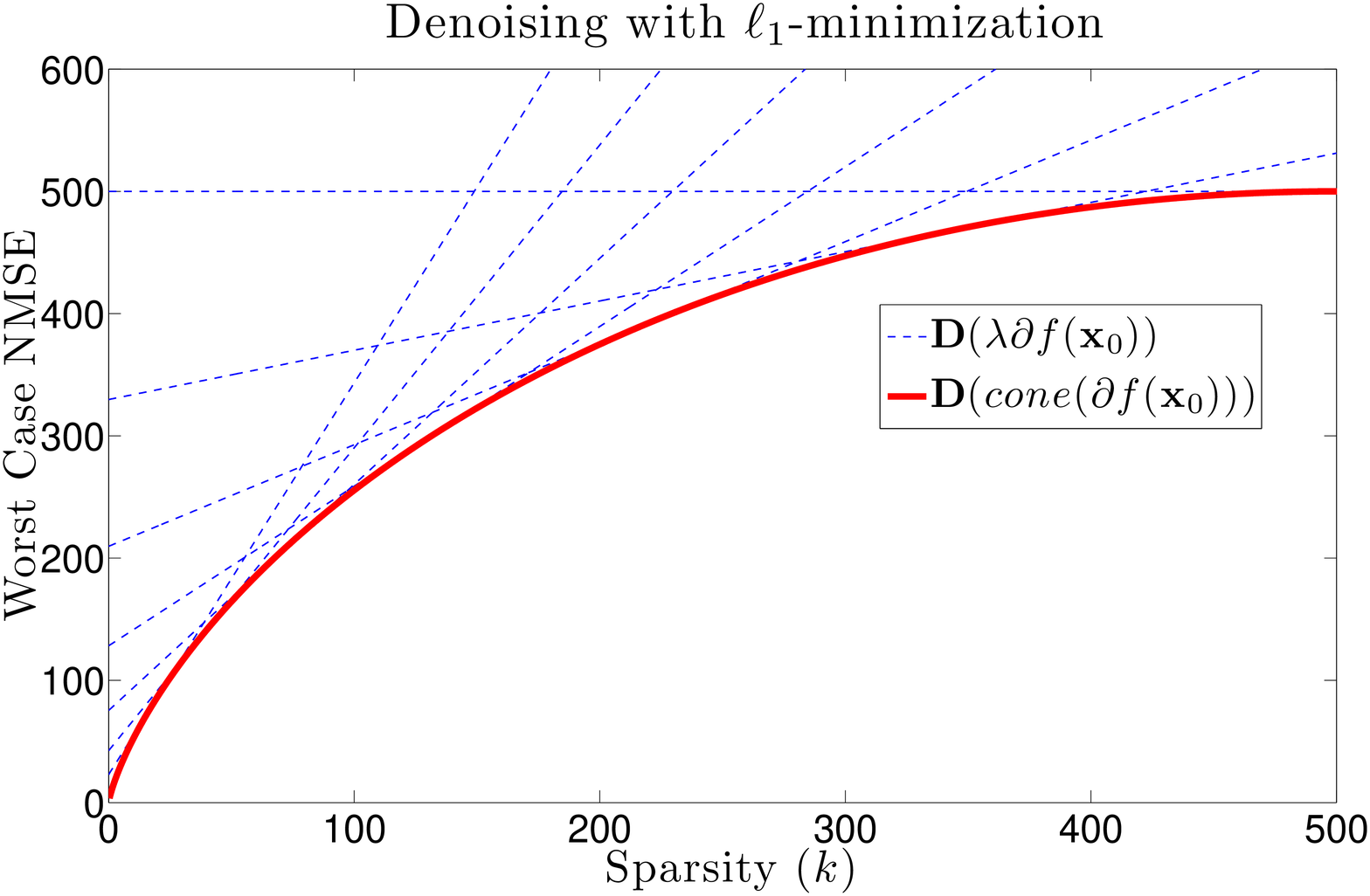}}  
  \end{center}
  \caption{{\small{We assumed $\x_0$ is a $k$ sparse vector in $\R^{500}$ and picked $f(\cdot)$ to be the $\ell_1$-norm. The solid red line corresponds to $\Delxf$ and dashed blue lines correspond to $\Dlf$ for $\la$ from $0$ to $1.5$. In case of $\ell_1$ minimization, both quantities are solely functions of the sparsity. Observe that $\Delxf$ is always upper bounded by $\Dlf$.}}}
 \label{ELL1Relation}
\end{figure}

\subsection{Connecting the NMSE to the linear inverse problem}\label{PT discussion}
We now turn our attention to the connection between the noiseless linear inverse problem and normalized MSE. Linear inverse problem is the problem of recovering a signal of size $n$ from $m$ noiseless linear observations $\y=\A\x_0$. To tackle this problem, when $m<n$, in a generic setup, a common approach is to make use of a structure inducing function $f(\cdot)$ and solving,
\beq
\min_\x f(\x)~~~\text{subject to}~~~\y=\A\x\label{LinInv}.
\eeq
It is known that, for structured signals, \eqref{LinInv} exhibits a \emph{phase transition} from failure to success as the number of observations increases. Works by Chandrasekaran et al. \cite{Cha} and Amelunxen et al. \cite{McCoy} showed that, this phase transition occurs around $m\approx \Delxf$. This is impressive as $\Delxf$ not only corresponds to the worst case NMSE of the constrained denoising \eqref{constrained prox}; but also to a seemingly unrelated problem \eqref{LinInv}. This observation was first made by Donoho et al. in \cite{Donoho}. Authors in \cite{Donoho} effectively claimed that, the worst case NMSE of optimally tuned \eqref{proxmain2} corresponds to the phase transition point of \eqref{LinInv}. We show that, this is indeed the case as $\min_{\la\geq 0} \Dlf\approx \Delxf$ (recall \eqref{sandwich}). Hence, our results combined with \cite{McCoy,Cha} rigorously justifies the claims in \cite{Donoho}.

\subsection{Relation to LASSO}
As a next step, we will now briefly extend our results to the estimation of $\x_0$ from noisy linear observations of $\x_0$. While many generic compressed sensing problems deal with matrices with independent identically distributed entries, we will be stating our results for a random partial unitary matrix $\A$. In other words, $\A$ is generated uniformly at random among the matrices satisfying $\A\in\R^{m\times n}$ and $\A\A^T=\Iden_{m}$.

Matrices with independent standard normal entries and uniformly random unitary matrices are clearly not identical; however, they are closely related. For instance, both matrices have uniformly distributed null spaces with respect to the Haar measure. This ensures that, the two matrices have statistically identical behavior for the purposes of the noiseless linear inverse problem \eqref{LinInv}, \cite{McCoy}.

We observe $\y=\A\x_0+\sigma\vb$ where $\vb\sim \Nn(0,\Iden)$. This is the standard setup for noisy linear inverse problems and it is used in numerous papers including \cite{Tikh,BicRit,BayMon,BunTsy,Mon,Koltc}. In order to estimate $\x_0$ from $\y$, we will be using the following generalized setup,
\begin{align}
\min_\x\|\y-\A\x\|_2^2~~~\text{subject to}~~~\x\in \Cc\label{LASSOopt}.
\end{align}
We can again relate this problem to the classic LASSO when $\x_0$ is sparse and when we choose $\Cc=\{\x\in\R^n~\big|~\|\x\|_1\leq \|\x_0\|_1\}$. When $f(\cdot)$ is an arbitrary function, we can use $\Cc=\{\x\in\R^n~\big|~f(\x)\leq f(\x_0)\}$.

We have the following result that relates optimal LASSO cost and the LASSO error to $\Delxc$ and also $\Delxf$.
\begin{thm}\label{LASSO} Assume $\Cc$ is a nonempty, closed and convex set and $\x_0\in\Cc$. Assume $\A\in\R^{m\times n}$ is a partial unitary matrix generated uniformly at random. Let the noise vector $\vb\sim\Nn(0,\Iden)$ be independent of $\A$ and $\y=\A\x_0+\sigma\vb$. Denote the minimizer of \eqref{LASSOopt} by $\x^*=\x^*(\y,\A)$. Conditioned on $\A$, define the normalized LASSO cost and ``projected LASSO error'' as follows,
\beq
F_{LASSO}(\A)=\lim_{\sigma\rightarrow 0}\frac{\E[\|\y-\A\x^*\|_2^2]}{\sigma^2},~~~\eta_{LASSO}(\A)=\max_{\sigma>0}\frac{\E[\|\A\x^*-\A\x_0\|_2^2]}{\sigma^2},\nn
\eeq
where the expectation is over $\vb$. Then, there exists constants $c_1,c_2>0$ such that,
\begin{itemize} 
\item Whenever $m< \Delxc$, with probability $1-c_1\exp(-c_2\frac{(m-\Delxc)^2}{n})$,
\beq
\eta_{LASSO}(\A)=m~~\text{and}~~F_{LASSO}(\A)=0.\nn
\eeq
\item Whenever $m>\Delxc$, with probability $1-c_1\exp(-c_2t^2)$,
\begin{align}
\Delxc-t\sqrt{n}\leq& \eta_{LASSO}(\A)\leq \Delxc+t\sqrt{n},\nn\\
m-\Delxc-t\sqrt{n}\leq& F_{LASSO}(\A)\leq m-\Delxc+t\sqrt{n}\nn
\end{align}
\end{itemize}
Here, the probabilities are over the random measurement matrix $\A$.
\end{thm}
This theorem suggests that, $\eta_{LASSO}$ and $F_{LASSO}$ has a phase transition around the point $m\approx \Delxc$; which not only shows up in \eqref{constprox} but also corresponds to the phase transition point of \eqref{LinInv} (setting $\Cc=\{\x\in\R^n~\big|~f(\x)\leq f(\x_0)\}$). When $m<\Delxc$, one cannot recover $\x_0$, even from the noiseless observations $\A\x_0$; as a result, it is futile to expect noise robustness. In this regime, our result suggests, $\eta_{LASSO}=m$. When the number of measurements are more than the phase transition point, interestingly, $\eta_{LASSO}$ stops growing proportionally to $m$ and takes a value around $\Delxc$ depending on the particular realization of $\A$.

While this is an interesting phenomenon, we should emphasize that, the more critical question is $\frac{\E[\|\x^*-\x_0\|_2^2]}{\sigma^2}$ rather than the projection of the error term on $\A$. This problem has been investigated for Gaussian measurement matrices and for $\ell_1$-minimization in a series of work, \cite{BayMon,Mon,StojLAS,StojSOCP}. However, to the best of our knowledge, Theorem \ref{LASSO} is the first result, that relates the LASSO cost and error to the convex geometry of the problem in a general framework.

\subsection{Organization of the paper}

The rest of this work is dedicated to the technical aspects of the theorems stated in this section. A summary of the consequent sections are as follows.

\vspace{3pt}

\noindent {$\bullet$ \bf{Section \ref{sec upbound}, Upper bounding the error:}} We will start by introducing the optimality conditions for the problem \eqref{prox}. Then, we will find a tight upper bound to the resulting MSE. This will prove Theorem \ref{con mini}.

\vspace{3pt}

\noindent {$\bullet$ \bf{Section \ref{lower for f}, Lower bounding the error:}} We will restrict our attention to the analysis of the proximity operator \eqref{proxmain2}. As $\sigma\rightarrow 0$, we will find a lower bound; which is arbitrarily close to the upper bound. This will prove Theorem \ref{proximal}. We use a similar argument for the constrained problem described in Theorem \ref{constrained}.

\vspace{3pt}

\noindent{$\bullet$ \bf{Section \ref{FurtherDis}, Further discussion}}: Following from \eqref{sandwich}, we discuss the relation between \emph{optimally tuned} (over $\la$) ``regularized estimator'' \eqref{proxmain2} and the ``constrained estimator'' \eqref{constrained prox}.

\vspace{3pt}

\noindent{$\bullet$ \bf{Section \ref{sec proj error}, Connection to LASSO:}} The proof of Theorem \ref{LASSO} will be provided. This proof requires some technical details related to the intersection of a random subspace and a cone. In particular, we will characterize certain properties of the intersection by using a modification of the results provided in \cite{McCoy}. We will also give examples to illustrate generality of our results.

\section{Upper Bound and Its Interpretation}\label{sec upbound}

For the rest of the paper, we will make the following assumptions.
\begin{itemize}
\item $f(\cdot):\R^n\rightarrow \R$ is a convex function.
\item $\Cc\subset\R^n$ is a nonempty, convex and closed set.
\item $\x_0\in\Cc$.
\end{itemize}

\noindent {\bf{Notation:}} When it is clear from context, the minimizer of an optimization over $\x$ will be denoted by $\x^*$. Following from Section \ref{notate1}, for a closed and convex set $\Cc$, the ``distance vector'' from $\Cc$ to $\x$ is defined as $\x-\bu(\x,\Cc)$ and is denoted by $\Pi(\x,\Cc)$. Since projection is the nearest point, $\|\Pi(\x,\Cc)\|_2=\dt(\x,\Cc)$.

To distinguish the deterministic analysis from stochastic analysis, we will allocate $\tau$ and solve,
\beq\label{con mini2}
\x^*=\arg\min_{\x\in\Cc}\tau f(\x)+\frac{1}{2}\|\y-\x\|_2^2
\eeq
when the noise is deterministic. When noise is distributed as $\Nn(0,\sigma^2\Iden)$, we will set $\tau=\sigma\la$ and $\la$ will correspond to the penalty for stochastic noise.

\subsection{Upper bounding the error}
\begin{lem} [Optimality conditions] \label{optimal lem} Let $\y=\x_0+\z$ and assume $f(\cdot)$ is a convex function and $\Cc$ is a convex and closed set containing $\x_0$. Let $\tau\geq 0$. Consider the problem \eqref{con mini2}. $\x^*$ is the unique optimal solution if and only if there exists $\s\in \pa f(\x^*)$ and $\ub\in T_\Cc(\x^*)^*$ such that:
\beq
\tau \s+\ub+\x^*=\y\label{KKT1}
\eeq
Alternatively, we can state \eqref{KKT1} in terms of the error vector $\w^*=\x^*-\x_0$ and the noise $\z$, as follows.
\beq
\tau \s+\ub+\w^*=\z\label{KKT2}
\eeq

\end{lem}
\begin{proof} The fact that ``$\x^*$ is the optimal solution if and only if \eqref{KKT1} holds'' follows from Proposition 4.7.2 of \cite{Bertse} and considering the subdifferential of the objective $\tau f(\x)+\frac{1}{2}\|\y-\x\|_2^2$. On the other hand, uniqueness follows from the strict convexity of the objective function of the proximity operator \eqref{con mini2}.
\end{proof}

The next lemma provides an upper bound on the $\ell_2$ norm of the error term $\w^*=\x^*-\x_0$.
\begin{lem}\label{upper bound lem} Consider the problem \eqref{con mini2}. Let $\w^*=\x^*-\x_0$. Assuming the setup of Lemma \ref{optimal lem}, we have:
\beq
\|\w^*\|_2\leq \dt(\z,\tau \pa f(\x_0)+T_\Cc(\x_0)^*).\nn
\eeq
\end{lem}
\begin{proof} Before starting the proof we emphasize that $\tau\paf+T_\Cc(\x_0)^*$ is closed due to the fact that $\tau\paf$ is compact and $T_\Cc(\x_0)^*$ is closed. From the optimality conditions, there exists $\s\in \pa f(\x^*)$ and $\ub\in T_\Cc(\x^*)^*$ such that \eqref{KKT2} holds. Now, let $\Kc:=\Kc(\Cc,f,\x_0)=\tau \pa f(\x_0)+T_\Cc(\x_0)^*$. Further, let $\s_0\in\pa f(\x_0)$, $\ub_0\in T_\Cc(\x_0)^*$ and $\w_0=\z-\tau\s_0-\ub_0$.
We will first explore the relation between the dual vectors and $\w^*$. The following inequality follows from the standard properties of the subgradients.
\beq
(\w^*)^T\s\geq f(\x_0+\w^*)-f(\x_0)\geq (\w^*)^T\s_0\implies (\w^*)^T(\s-\s_0)\geq 0\label{my sub eq}
\eeq
Since $\ub\in T_\Cc(\x^*)^*$ and $-\w^*\in T_\Cc(\x^*)$, we have:
\beq
\li\ub,\w^*\ri\geq 0\label{my w eq1}
\eeq
Similarly, $\ub_0\in T_\Cc(\x_0)^*$ and $\w^*\in T_\Cc(\x_0)$, hence:
\beq
\li\ub_0,\w^*\ri\leq 0\label{my w eq2}
\eeq
Overall, combining \eqref{my sub eq}, \eqref{my w eq1} and \eqref{my w eq2}, we find:
\beq
\li\w^*,(\ub+\tau\s)-(\ub_0+\tau\s_0)\ri\geq 0\label{importw}
\eeq
From \eqref{importw}, we will conclude that $\|\w^*\|_2\leq \|\w_0\|_2$. \eqref{importw} is equivalent to:
\beq
\li\w^*,(\z-\w^*)-(\z-\w_0)\ri=\li\w^*,\w_0-\w^*\ri\geq 0\implies \|\w^*\|_2^2\leq  \li\w^*,\w_0\ri\leq \|\w^*\|_2\|\w_0\|_2\label{importw2}
\eeq
Hence, we indeed have: $\|\w_0\|_2\geq \|\w^*\|_2$. Since this is true for all $\s_0,\ub_0$, we can choose the $\w_0$ with the shortest length, namely, choose $\w_0$ to be the distance vector $\w_0=\Pi(\z,{T_\Cc(\x_0)^*+\tau\pa f(\x_0)})$ to conclude.
\end{proof}

As a direct application of Lemma \ref{upper bound lem}, we can upper bound the expected errors as follows when $\z$ is stochastic.
\begin{propo} \label{upper summary} Let $\sigma>0$, $\la\geq 0$. In problem \eqref{con mini2}, let $\z=\sigma\vb$,  where $\vb\sim\Nn(0,\Iden)$. Then,
\begin{itemize}
\item Denote the minimizer of \eqref{con mini2} by $\x^*=\x^*(\tau,\z)$. Letting $\tau=\sigma\la$,
\beq
\frac{\E[\|\x^*-\x_0\|_2^2]}{\sigma^2}\leq \Dctf.\label{always upper}
\eeq
\item Let $\la^*=\arg\min_{\la\geq 0}\Dctf\label{tau^* is good}$. By choosing $\tau=\la^*\sigma$ in \eqref{con mini2} and using \eqref{always upper} we have,
\beq
\frac{\E[\|\x^*-\x_0\|_2^2]}{\sigma^2}\leq \min_{\la\geq 0} \Dctf\label{good estimate}
\eeq
\item Setting $\Cc=\R^n$ in \eqref{always upper}, we show $\leq $-direction of Theorem \ref{proximal}.
\item Setting $f(\cdot)=0$ in \eqref{always upper}, we show $\leq$-direction of Theorem \ref{constrained}.
\end{itemize}
\end{propo}
\begin{proof}
Let $\tau=\sigma\la$ and $\w^*(\tau,\z)=\x^*(\tau,\z)-\x_0$. From Lemma \ref{upper bound lem}, the noise in \eqref{con mini2} satisfies:
\beq
\|\w^*(\tau,\z)\|_2\leq \dt(\z,\sigma\la\pa f(\x_0)+T_\Cc(\x_0)^*)=\dt(\sigma\vb,\sigma\la\pa f(\x_0)+T_\Cc(\x_0)^*).\nn
\eeq
Taking the squares of both sides and normalizing by $\sigma^2$, we find,
\beq
\frac{\|\w^*(\tau,\z)\|_2^2}{\sigma^2}\leq \dt(\vb,\la\pa f(\x_0)+T_\Cc(\x_0)^*)^2.\nn
\eeq
Now, taking the expectations of both sides and using $\vb\sim\Nn(0,\Iden)$, we end up with the first statement \eqref{always upper}. Second statement follows immediately. For third statement, observe that if $\Cc=\R^n$ (i.e. no constraint) then $F_\Cc(\x_0)=\R^n$ and $T_\Cc(\x_0)^*=\{0\}$. Finally, if $f(\cdot)=0$ (i.e. no penalization), $\paf=\la\paf=\{0\}$.

%
%
\end{proof}

\subsection{First order approximation}
In this section, we will consider the first order approximation of the proximal denoising problem around $\x_0$ and we will argue that, the error term of the first order problem is exactly same as the upper bound provided in Lemma \ref{upper bound lem}. This will provide a good intuition for our consecutive results, which argue that the least favorable noise distribution that maximizes normalized MSE occurs as $\sigma\rightarrow 0$.\\
{\bf{Notation:}} Given $\w,\x_0,f(\cdot),\eps>0$, the ratio $\frac{f(\x_0+\eps\w)-f(\x_0)}{\eps}$ is nondecreasing as a function of $\eps$. Furthermore, the ``directional derivative'' along $\w$ is defined as \cite{Roc70,Urru},
\beq
\lim_{\eps\rightarrow0} \frac{f(\x_0+\eps\w)-f(\x_0)}{\eps}=f'(\x_0,\w)\label{dirder3}
\eeq
It is known that, \cite{Urru}, there is a trivial relation between the directional derivative and the subgradients,
\beq
f'(\x_0,\w)=\sup_{\s\in\paf}\li\w,\s\ri.\label{dirder}
\eeq
We will extensively make use of this equality later on. Finally, let us denote the ``set of maximizing subgradients'' of $f(\cdot)$ at $\x_0$ along $\w$, via $\pafw$. In particular,
\beq
\pafw:=\{\s\in\paf\big|\s^T\w=\sup_{\s'\in\paf}(\s')^T\w\}\label{dir dev2}
\eeq
$\pafw$ basically corresponds to a face of the subdifferential along $\w$. Observe that, the supremum on the right hand side is the directional derivative of $f(\cdot)$ at $\x_0$ along $\w$.

\subsubsection{Approximated problem}
\begin{itemize}
\item {\bf{Approximating the function $f(\cdot)$ at $\x_0$:}} We will use the first order approximation $\ff(\cdot)$ of $f(\cdot)$ at $\x_0$, which is achieved via directional derivatives. In particular, we will let:
\beq
\ff(\x)=f(\x_0)+\sup_{\s\in \pa f(\x_0)}\li\s,\x-\x_0\ri\label{f App}
\eeq
It is known that, $f(\cdot)$ behaves like $\ff(\cdot)$ in the close proximity of $\x_0$, \cite{Bertse}. \eqref{f App} can be viewed as a generalization of the first order approximation of differentiable functions. Furthermore, observe that, $\ff(\x)\leq f(\x)$ for all $\x$ due to convexity and the subgradient property.
\item {\bf{Approximating the set $\Cc$ at $\x_0$:}} We will find an approximation of $\Cc$ by considering the tangent cone of $\Cc$ at $\x_0$. This follows from the fact that $T_\Cc(\x_0)$ provides a good approximation of the set of feasible perturbations $F_\Cc(\x_0)$ around $0$ as we have $T_\Cc(\x_0)=\text{Cl}(\cn(F_\Cc(\x_0)))$. Hence, the approximation to the set $\Cc$ will be:
\beq
\hat{\Cc}_{\x_0}=\x_0+T_\Cc(\x_0)\label{c App}
\eeq
\item {\bf{The approximated problem around $\x_0$:}} Approximated problem is obtained by approximating the function and the set simultaneously and considering \eqref{con mini2} with the new function and set $\hat{f},\hat{\Cc}$.
\beq
\min_{\x\in\hat{\Cc}_{\x_0}} \tau\ff(\x)+\frac{1}{2}\|\y-\x\|_2^2\label{first or}
\eeq
As it is argued in Lemma \ref{approx conv}, the first order approximation of $f(\cdot)$ is a convex function and hence \eqref{first or} is a convex problem.
\end{itemize}

The nice property of the problem \eqref{first or} is the fact that, it can be solved exactly. It also provides an insight about what to expect as the minimax risk. We will now show that, the error in the approximated problem is exactly equal to the upper bound given in Lemma \ref{upper bound lem}. Hence, the error in the original problem is always upper bounded by the error in the approximated problem. However, intuitively, the approximation will be tight when the noise variance is sufficiently small, hence when the variance $\sigma$ is small, the error in the original problem will be approximately equal to the upper bound.

Before analysis of the approximated problem, we will provide a lemma which shows that the subgradients of $\ff$ can be found exactly. The proof can be found in the Appendix \ref{approximate subdif}.
\begin{lem}\label{approx subdif} Let $\ff$ be same as in \eqref{f App}. For any $\x$, we have,
\begin{align}
&\pa \ff(\x)=\pafb\x-\x_0)\label{dir dev}
\end{align}
\end{lem}
Basically, the subgradients of $\ff$ are the maximizing subgradients of $f(\cdot)$, i.e., the set of subgradients that maximizes the inner product with $\x-\x_0$.

The next result, characterizes a useful property of the projection on a convex set, \cite{Bertse}.
\begin{lem} \label{proj maximize} Let $\Kc$ be a closed and convex set and $\vb$ be an arbitrary vector. We have,
\beq
\li\bu(\vb,\Kc),\Pi(\vb,\Kc)\ri=\sup_{\ub\in\Kc}\li\ub,\Pi(\vb,\Kc)\ri\nn
\eeq
In words, the projection vector maximizes the inner product with the distance vector over $\Kc$.
\end{lem}

We have the following result regarding the first order approximated problem \eqref{first or}.
\begin{propo} [Solution for the approximated problem]\label{approx prop} Let $\ff$ and $\hat{\Cc}$ be same as \eqref{f App} and \eqref{c App} respectively. Let $\hax=\hat{\x}(\tau,\z)$ denote the optimal solution of \eqref{first or}. Then, we have:
\beq\label{approx}
\hax-\x_0=\Pi(\z,{\tau \pa f(\x_0)+T_\Cc(\x_0)^*})
\eeq
\end{propo}
\begin{proof}
Let $\haw=\hax-\x_0$. From optimality conditions, $\hax$ is the unique minimizer of \eqref{first or} if and only if, it satisfies,
\beq
\tau \s+\ub+\haw=\z\label{satisfied}
\eeq
where $\s\in \pa \ff(\hax)$ and $\ub\in T_{\hat{\Cc}_{\x_0}}(\hax)^*$. Using Lemma \ref{approx subdif}, $\s\in \pa \ff(\hax)$ is identical to $\s\in \pafb\hax-\x_0)$.

Assume \eqref{approx} holds, then $\z-\haw=\tau\s+\ub$ for some $\ub\in T_\Cc(\x_0)^*$ and $\s\in \paf$. We will show that, \eqref{satisfied} holds for this particular $\haw,\s,\ub$ and $\hax$ is indeed the minimizer of \eqref{first or}. Since $\z-\ub$ is the closest point to $\tau \paf$, we have:
\beq
\haw=\Pi(\z-\ub,\tau \pa f(\x_0)),~ \tau\s=\bu(\z-\ub,\tau \pa f(\x_0))\nn
\eeq
In conjunction with Lemma \ref{proj maximize}, this implies $\li\haw,\s\ri=\sup_{\s'\in\pa f(\x_0)}\li\s',\haw\ri$ which in turn implies $\s\in \pafb\hax-\x_0)= \pa\ff(\hax)$ based on Lemma \ref{approx subdif}. In a similar manner, we have:
\beq
\haw=\Pi(\z-\tau\s,T_\Cc(\x_0)^*),~ \ub=\bu(\z-\tau\s,T_\Cc(\x_0)^*)\nn
\eeq
Using Moreau's decomposition theorem (see Lemma \ref{more}), this implies that $\haw\in T_\Cc(\x_0)$ and $\li\haw,\ub\ri=0$. Now, we will argue that $\ub\in T_{\hat{\Cc}_{\x_0}}(\hax)^*$ to conclude. Let $\vb\in F_{\hat{\Cc}_{\x_0}}(\hax)$. This implies $\vb+\x_0+\haw\in \hat{\Cc}_{\x_0}$ or alternatively, $\vb+\haw\in T_\Cc(\x_0)$. Hence,
\beq
\li\vb+\haw,\ub\ri\leq 0\nn
\eeq
Using $\li\haw,\ub\ri=0$, we find for all $\vb\in F_{\hat{\Cc}_{\x_0}}(\hax)$, $\li\vb,\ub\ri\leq0$. This will still be true when we take the closure and the cone of the set $F_{\hat{\Cc}_{\x_0}}(\hat{\x}(\tau,\z))$ which gives $\li\vb,\ub\ri\leq0$ for all $\vb\in T_{\hat{\Cc}_{\x_0}}(\hax)$ and implies $\ub$ is in the polar cone of $T_{\hat{\Cc}_{\x_0}}(\hax)$.
\end{proof}
This shows that the approximated problem \eqref{first or} can be solved exactly and its error exactly matches with the upper bound found in Lemma \ref{upper bound lem}. This observation provides a good interpretation of Lemma \ref{upper bound lem}.

The next corollary follows immediately from Proposition \ref{approx prop} and gives MSE results under stochastic noise. The proof is based on taking the expectation of \eqref{approx} over $\z$.
\begin{cor}[Approximated MSE]\label{approx risk} Assume $\vb\sim\Nn(0,\Iden)$. Let $\z=\sigma\vb$ in the approximated problem \eqref{first or} and let $\tau=\sigma\la$. Denote the minimizer by $\hat{\x}(\la,\z)$. Then,

\beq\E[\|\hat{\x}(\la,\z)-\x_0\|_2^2]=\sigma^2\E[\dt(\vb,\la \pa f(\x_0)+T_\Cc(\x_0)^*)^2].\nn\eeq
\end{cor}

The next section is dedicated to showing that the bounds we have provided so far are in fact sharp. In particular, as $\sigma\rightarrow 0$, the original problem \eqref{con mini2} behaves same as the approximated problem \eqref{first or}.

\section{Tight lower bound when $\sigma\rightarrow 0$}\label{lower for f}
We will now focus our attention to the $\geq$-direction of Theorem \ref{proximal} and hence the regularized estimator (no constraint),
\beq
\x^*(\tau,\z)=\arg\min \tau f(\x)+\frac{1}{2}\|(\x_0+\z)-\x\|_2^2\label{penest}
\eeq
We will show that, one can find a good lower bound for the estimation error $\|\x^*(\tau,\z)-\x_0\|_2$ in \eqref{penest}. The following proposition states that $f(\cdot)$ can be approximated by $\ff$(recall \eqref{f App}) around $\x_0$ along all directions simultaneously. This result is a stronger version of \eqref{dirder} and is due to Lemma 2.1.1 of Chapter VI of \cite{Urru}.

\begin{propo}\label{unif lem} Assume $f(\cdot)$ is a convex function on $\R^n$. Then, for any given $\delta_0>0$, there exists some $\eps_0>0$ such that, for all $\w\in\R^n$, with $\|\w\|_2\leq \eps_0$, the following holds:
\beq
0\leq f(\x_0+\w)-\ff(\x_0+\w)\leq \delta_0\|\w\|_2\label{relation eps delt}
\eeq
\end{propo}

While it is clear that $\ff(\cdot)$ is an under estimator of $f(\cdot)$ at all points, upper bounding $f(\cdot)$ in terms of $\ff(\cdot)$ and a small first order term will be quite helpful for the perturbation analysis. 

\subsection{Deterministic lower bound}

The following lemma provides a deterministic lower bound on the error $\x^*(\tau,\z)-\x_0$. Unlike the previous sections, there will be no set constrained (i.e. $\Cc$) in the consequent analysis.
\begin{lem}\label{sharp} Consider the problem \eqref{penest}. Let $\w_0:=\Pi(\z,\tau\paf)$. Assume $\tau\leq C_0\|\w_0\|_2$ for some constant $C_0>0$. Using Proposition \ref{unif lem}, choose $\delta_0,\eps_0>0$ ensuring \eqref{relation eps delt}. Then, whenever $\|\w_0\|_2\leq \eps_0$,
\beq
\|\x^*(\tau,\z)-\x_0\|_2\geq (1-\sqrt{2\delta_0C_0} )\dt(\z,\tau\paf).\nn
\eeq
\end{lem}
\begin{proof}
To begin with, observe that, $\|\w_0\|_2=\dt(\z,\tau\paf)$. Let $\s_0=\frac{\bu(\z,\tau\paf)}{\tau}$. From Lemma \ref{proj maximize}, $\s_0\in \pafb\w_0)$ where $\pafb\w_0)$ is given by \eqref{dir dev2}. Let us rewrite the objective function of (\ref{penest}) as a function of the perturbation $\w=\x-\x_0$.
\beq
g(\w)=\tau (f(\x_0+\w)-f(\x_0))+\frac{1}{2}\|\z-\w\|_2^2\nn
\eeq
Using the facts that $\|\w_0\|_2\leq\eps_0$, $\tau\leq C_0\|\w_0\|_2$ and $\w_0+\tau\s_0=\z$, we have the following,
\bea
\nn g(\w_0)&=\tau (f(\x_0+\w_0)-f(\x_0))+\frac{1}{2}\|\z-\w_0\|_2^2\\
\nn&=\tau (\li\s_0,\w_0\ri+(f(\x_0+\w_0)-\hat{f}_{\x_0}(\x_0+\w_0)))+\frac{1}{2}\|\z-\w_0\|_2^2\\
\nn&\leq \tau (\li\s_0,\w_0\ri+\delta_0\|\w_0\|_2)+\frac{1}{2}\|\w_0\|_2^2-\li\z,\w_0\ri+\frac{1}{2}\|\z\|_2^2\\
\nn&\leq  C_0 \delta_0\|\w_0\|_2^2+\frac{1}{2}\|\w_0\|_2^2-\li\z-\tau\s_0,\w_0\ri+\frac{1}{2}\|\z\|_2^2\\
\nn&= (C_0 \delta_0-\frac{1}{2})\|\w_0\|_2^2+\frac{1}{2}\|\z\|_2^2\\
\label{eqq}&= (C_0 \delta_0-\frac{1}{2})\dt(\z,\tau\paf)^2+\frac{1}{2}\|\z\|_2^2.
\end{align}
Letting $\w^*(\tau,\z)=\x^*(\tau,\z)-\x_0$, $g(\w^*(\tau,\z))$ should be smaller than the right hand side of (\ref{eqq}) as $\x^*(\tau,\z)$ should outperform $\x_0+\w_0$. Choose $\s^*\in \pafb\x^*(\tau,\z)-\x_0)$.
Then, $g(\w^*(\tau,\z))$ can be lower bounded as follows:
\bea
g(\w^*(\tau,\z))&\geq \tau\li\s^*,\w^*(\tau,\z)\ri+\frac{1}{2}\|\z-\w^*(\tau,\z)\|_2^2\nn\\
&=\tau \li\s^*,\w^*(\tau,\z)\ri-\li\z,\w^*(\tau,\z)\ri+\frac{1}{2}\|\w^*(\tau,\z)\|_2^2+\frac{1}{2}\|\z\|_2^2\nn\\
&= -\li\w^*(\tau,\z),\z-\tau\s^*\ri+\frac{1}{2}\|\w^*(\tau,\z)\|_2^2+\frac{1}{2}\|\z\|_2^2\label{bound on gw}
\end{align}
Now, using $\s^*\in \pafb\w^*(\tau,\z))$,
\beq
\li\w^*(\tau,\z),\s^*\ri\geq \li\w^*(\tau,\z),\s_0\ri\implies -\li\w^*(\tau,\z),\z-\tau\s^*\ri\geq -\li\w^*(\tau,\z),\z-\tau\s_0\ri\nn
\eeq
Combining this and \eqref{bound on gw},
 $g(\w^*(\tau,\z))$ satisfies:
\beq\label{func1}
g(\w^*(\tau,\z))\geq \frac{1}{2}\|\w^*(\tau,\z)\|_2^2-\|\w^*(\tau,\z)\|_2\dt(\z,\tau\paf)+\frac{1}{2}\|\z\|_2^2
\eeq
From Lemma \ref{upper bound lem}, it is known that $\|\w^*(\tau,\z)\|_2\leq \dt(\z,\tau\paf)$. Now, calling $a=\|\w^*(\tau,\z)\|_2,~b=\dt(\z,\tau\paf)$, and $t=2C_0 \delta_0$, $g(\w^*(\tau,\z))\leq g(\w_0)$ implies,
\beq
\frac{a^2}{2}-ab\leq (\frac{t-1}{2})b^2.\nn
\eeq
Setting $u=a/b$ and using $1\geq u,t\geq 0$, we have:
\beq
ub^2-u^2b^2/2\geq (1-t)b^2/2\implies 2u-u^2\geq (1-t)\implies u^2-2u+(1-t)\leq 0\nn
\eeq
Overall, $1\geq u\geq 1-\sqrt{t}$. Equivalently, $\|\w^*(\tau,\z)\|_2\geq (1-\sqrt{2C_0 \delta_0})\dt(\z,\tau\paf)$.
\end{proof}

This lemma shows that, when $\z$ is sufficiently small, then the error $\|\w^*\|_2$ will be close to its upper bound given in Lemma \ref{upper bound lem}. Observe that, we assumed $\tau\leq C_0\dt(\z,\tau \paf)$ for some $C_0>0$. The setup of Theorem \ref{proximal} already prepares the background for this assumption as it chooses $\tau=\sigma\la$ and $\z\sim\Nn(0,\sigma^2\Iden)$. Together, these will ensure $\z,\tau$ is approximately proportional and Lemma \ref{upper bound lem} is applicable, with high probability, for sufficiently large $C_0$. The choice $\tau=\sigma\la$ for a fixed $\la$, is in fact a well-known property of the soft thresholding operator, \cite{soft-thresh,Donoho,Mon,Matesh,Tikh,BicRit}.


\subsection{Stochastic lower bound for regularized NMSE}
To finish the proof of Theorem \ref{proximal}, we need to show that the upper bound can be achieved as $\sigma\rightarrow 0$. This will be the topic of this section.

\begin{propo} \label{asymp}Consider program \eqref{penest} with $\tau=\sigma\la$ and $\z=\sigma\vb$ where $\vb\sim\Nn(0,\Iden_n)$. Then, we have:
\beq
\lim_{\sigma\rightarrow 0}\frac{\E[\|\x^*(\sigma\la,\sigma\vb)-\x\|_2^2]}{\sigma^2}=\Dlf\nn
\eeq
\end{propo}
\begin{proof} To begin with, observe that $\Dlf>0$ as $\la\paf$ is a compact set. Let $\w^*=\x^*(\sigma\la,\sigma\vb)-\x$. Denote the probability density function of an $\Nn(0,c\Iden)$ distributed vector by $p_c(\cdot)$. We can write:
\beq
\E[\|\w^*\|_2^2]=\int_{\z\in\R^n}\| \w^*\|_2^2 p_\sigma(\z)d\z\nn
\eeq
We will split the error term into three parts where the integration is performed over the following sets.
\begin{itemize}
\item $S_1=\{\ab\in\R^n\big|\|\ab\|_2\geq C_1\sigma\}$.
\item $S_2=\{\ab\in\R^n\big| \|\ab\|_2\leq C_1\sigma,~\dt(\ab,\la\sigma \paf)\geq \eps_1\sigma\}$.
\item $S_3=\{\ab\in\R^n\big| \|\ab\|_2\leq C_1\sigma,~\dt(\ab,\la\sigma \paf)\leq \eps_1\sigma\}$.
\end{itemize}
where $C_1,\eps_1$ are positive scalars to be determined later on.
Recall that $\vb=\frac{\z}{\sigma}$ and define the associated restricted integrations $I_i,\hat{I}_i$, which are given as:
\begin{align}
I_i&=\int_{\z\in S_i} \|\w^*\|_2^2 p_\sigma(\z)d\z\nn\\
\hat{I}_i&=\int_{\sigma\vb\in S_i} \dt(\vb,\la\paf)^2 p_1(\vb)d\z\nn
\end{align}
From Proposition \ref{upper summary}, $\E[\|\w^*(\sigma\la,\sigma\vb)\|_2^2]=I_1+I_2+I_3\leq\sigma^2( \hat{I}_1+\hat{I}_2+\hat{I}_3)$. To proceed,
\begin{itemize}
\item We will first argue that $\hat{I_1}$ and $\hat{I_3}$ are smaller compared to $\hat{I_2}$. 
\item Then, we will argue that $I_2$ is close to $\hat{I}_2$.
\item These will yield: 
\beq
I_1+I_2+I_3\geq I_2\approx \hat{I}_2\approx \hat{I}_1+\hat{I}_2+\hat{I}_3\nn
\eeq
\end{itemize}

\noindent{\bf{Claim:}} We have,
\beq
\hat{I}_2\geq\Dlf-[\eps_1^2+Q(C_1)]\label{I2eq}
\eeq
where $Q(C_1)\rightarrow 0$ as $C_1\rightarrow\infty$.
\begin{proof} Using the boundedness of $\paf$, let $R=R_{\x_0}=\sup_{\s\in\paf}\|\s\|_2$. $\dt(\vb,\la\paf)^2$ can be upper bounded as,
\beq
\dt(\vb,\la\paf)^2\leq (\|\vb\|_2+\la R)^2\leq 2(\|\vb\|_2^2+\la^2R^2)\nn.
\eeq
Observe that, $\hat{I}_1$ is calculated over the region $\|\vb\|_2\geq C_1$. As $C_1\rightarrow \infty$, since the multivariate Gaussian distribution have finite moments, we have $\hat{I}_1:=Q(C_1)\rightarrow 0$. 

$\hat{I}_3$ can be upper bounded by the integration of $\eps_1^2$ over $\R^n$ which yields $\hat{I}_3\leq \eps_1^2$. Finally, using,
\beq
\sum_{i=1}^3\hat{I}_i=\Dlf\nn
\eeq
we have $\hat{I_2}\geq \Dlf-\eps_1^2-Q(C_1)$.
\end{proof}
The following lemma finishes the proof by using Lemma \ref{sharp}.
\begin{lem} Let $R=\sup_{\s\in\paf}\|\s\|_2$. Based on Proposition \ref{unif lem}, for any $\delta_0>0$, there exists $\eps_0>0$ so that,
\beq
f(\x_0+\w)-\hat{f}_{\x_0}(\x)\leq \delta_0\|\w\|_2~\text{for all}~\|\w\|_2\leq \eps_0\label{cool delta eps relation}.
\eeq
Now, for any given (fixed) $\eps_1$, $C_1$, $\delta_0$, $\eps_0$, choose $\sigma$ sufficiently small to ensure $(C_1+R\la)\sigma<\eps_0$. Then, for all such $\sigma$'s,
\beq
I_2\geq \sigma^2(1-\sqrt{2\delta_0\la\eps_1^{-1}})\hat{I_2}\label{I2comp}.
\eeq

Following from \eqref{I2comp} and letting $\delta_0\rightarrow 0$ (by choosing $\eps_0$ properly) and keeping $\sigma$ accordingly small, we obtain,
\beq
\lim_{\sigma\rightarrow 0}\frac{\E[\|\w^*(\la,\z)\|_2^2}{\sigma^2 \hat{I_2}}\geq \lim_{\sigma\rightarrow 0}\frac{I_2}{\sigma^2 \hat{I_2}}\geq 1\label{I2comp2}.
\eeq
Since this is true for arbitrary $\eps_1,C_1$ letting $C_1\rightarrow \infty$ and $\eps_1\rightarrow 0$ and using the relation \eqref{I2eq}, we can conclude,
\beq
\lim_{\sigma\rightarrow 0}\frac{\E[\|\w^*(\la,\z)\|_2^2]}{\sigma^2}\geq \Dlf\label{sigma-0}.
\eeq
\end{lem}
\begin{proof}
To show the result, we will use Lemma \ref{sharp}. To apply Lemma \ref{sharp} with $\w_0=\dt(\z,\sigma\la\paf)$, we will show that, all $\z\in S_2$ satisfy the requirements. 
\begin{itemize}
\item First, observe that, for any $\z\in S_2$, from triangle inequality, we have,
\beq
\dt(\z,\sigma\la\paf)\leq \sigma\la R+\sigma C_1\leq (\la R+C_1)\sigma\leq \eps_0.\nn
\eeq
\item Secondly, we additionally have,
\beq
\tau=\sigma\la\leq \frac{\la}{\eps_1}\dt(\z,\sigma\la\paf).\nn
\eeq
\item Using these and $(C_1+\la R)\sigma<\eps_0$, we can apply Lemma \ref{sharp} for all $\z\in S_2$ to obtain,
\beq
\|\w^*(\la,\z)\|_2\geq(1-\sqrt{2\delta_0 \la\eps_1^{-1}})\dt(\z,\sigma\la\paf)\label{final final bound}
\eeq
\end{itemize}
Integrating over $S_2$, we find \eqref{I2comp}. Recall that, $\sigma$ is under our control and all $\z\in S_2=S_2(\sigma)$ with $\sigma<\frac{\eps_0}{\la R+C_1}$ satisfies \eqref{final final bound}; which is only a function of $\delta_0$ when $C_1,\eps_1$ are fixed. We can let $\delta_0\rightarrow 0$, while keeping $\sigma$ sufficiently small compared to $\eps_0$, to obtain, \eqref{I2comp2} after an integration over $S_2$.

\eqref{I2comp2} is true for any constants $C_1,\eps_1>0$ as $\sigma\rightarrow 0$. Letting $C_1\rightarrow \infty, \eps_1\rightarrow 0$ and using \eqref{I2eq}, we obtain,
\beq
\lim_{C_1\rightarrow\infty,\eps_1\rightarrow 0}\frac{\hat{I}_2}{\Dlf}= 1.\nn
\eeq
Finally, combining this with \eqref{I2comp2}, and letting $\sigma\rightarrow 0$ as a function of $C_1\rightarrow\infty,\eps_1\rightarrow 0$, we get the desired result \eqref{sigma-0}. Hence, this shows the $\geq $--direction in Theorem \ref{proximal} as $\sigma\rightarrow 0$.
\end{proof}

\end{proof}

\section{Further discussion on proximal denoising}\label{FurtherDis}
The remaining two things are the proofs for Theorem \ref{constrained} and discussion of \eqref{sandwich} which relates Theorem \ref{constrained} and \ref{proximal}. We should remark that, for Theorem \ref{constrained}, Proposition \ref{upper summary} already provides the upper bound on the estimation error. In order to show the exact equality in \eqref{constrained eq}, we use arguments which are similar to Proposition \ref{asymp} in nature. This time, instead of arguing the function $f(\cdot)$ can be well approximated by the directional derivatives in a small neighborhood, we consider the first order approximation of $F_\Cc(\x_0)$ by the tangent cone $T_\Cc(\x_0)$. The reader is referred to Appendix \ref{constrained case} for the proof.

We will now discuss the relation between the terms $\Delxf$ and $\Dlf$; which relates constrained and regularized estimation results and also the noiseless linear inverse problem \eqref{LinInv} to each other.
From Propositions \ref{asymp} and \ref{upper summary}, recall that $\Dlf$ corresponds to the worst case NMSE of the $\la$-regularized proximity operator \eqref{proxmain2}. When we tune $\la$ optimally, we can achieve an error as small as $\inf_{\la\geq 0}\Dlf$. We related this to $\Delxf$ via \eqref{sandwich}. The right hand side is due to Theorem $4.5$ of \cite{McCoy} which gives the following bound,
\beq
\inf_{\la\geq 0}\Dlf\leq \Delxf+\frac{2\sup_{\s\in\paf}\|\s\|_2}{f(\frac{\x_0}{\|\x_0\|_2})}\label{technical}
\eeq
when $f(\cdot)$ is a norm. This bound can be enhanced by choosing an $\x$ that maximizes $f(\frac{\x}{\|\x\|_2})$ while ensuring $\pa f(\x)=\paf$. We will now provide some observations on \eqref{technical}. Later on, we will exemplify how this upper bound is quite simple and powerful for well-known structure inducing functions.  We will also propose an alternative upper bound of our own that only depends on the subdifferential $\paf$ and does not require $f(\cdot)$ to be norm.

\subsection{Observations on the upper bound}
Upper bounding $\inf_{\tau\geq 0}\Dx\tau f)$ in terms of $\Delxf$ requires technical argument involving Gaussian concentration. While \eqref{technical} looks simple, magnitudes of $\sup_{\s\in\paf}\|\s\|_2$ and $f(\frac{\x_0}{\|\x_0\|_2})$ might not always be clear.

For the consequent discussion, we will use the following notation. Let,
\beq
R_{\x_0}=\sup_{\s\in\paf}\|\s\|_2~~~~~\text{and}~~~~~f_{\max}(\x_0)=\sup_{\pa f(\x)=\pa f(\x_0)}f\left(\frac{\x}{\|\x\|_2}\right).\eeq $f_{\max}(\x_0)$ corresponds to the largest value of $f(\frac{\x}{\|\x\|_2})$ while keeping subdifferential fixed. Since $\inf_{\tau\geq 0}\Dx\tau f)$ and $\Delxf$ depends only on the set subdifferential we can use $f_{\max}(\x_0)$ to get a better bound on the right hand side of \eqref{technical}.

We will now briefly investigate this quantity for the classical sparsity inducing functions and argue that it has a simple interpretation in general.
\begin{itemize}
\item{$\ell_1$ norm:} If $\x_0$ is a $k$ sparse signal, then, subgradient is given as:
\beq
\pa\|\x_0\|_1=\left\{\s\in\R^n\big |\begin{cases}s_i=\text{sgn}(x_{0,i})~\text{when }~x_{0,i}\neq 0\\s_i\in [-1,1]~\text{else}\end{cases}\right\}
\eeq
Hence, subgradients only depend on the sign of the signal. This way, by keeping $\text{sgn}(\x_0)$ same and normalizing the magnitudes, we can make $f(\frac{\x_0}{\|\x_0\|_2})$ as large as $\sqrt{k}$. On the other hand, $R_{\x_0}$ is trivially equal to $\sqrt{n}$ and it is achieved by choosing $|s_i|=1$ for all $1\leq i\leq n$. Consequently, we find:
\beq
\frac{R_{\x_0}}{f_{\max}(\x_0)}=\sqrt{\frac{n}{k}}\label{quant 1}
\eeq
\item{Nuclear norm:} We now assume $\x_0$ is a vectorized form of a rank $r$ matrix $\X_0\in\R^{d\times d}$ where $n=d^2$. Assume, $\X_0$ has singular value decomposition $\U\Sigma\V^T$ where $\Sigma\in\R^{r\times r}$ has positive diagonals. Set of subgradients of the nuclear norm at $\X_0$ can be given as \cite{Decomp1},
\beq
\pa \|\X_0\|_\st =\{\Sb\big|\U^T\Sb\V=\Iden_{r}~~~\text{and}~~~\|(\Iden_{d}-\U\U^T)\Sb(\Iden_{d}-\V\V^T)\|_\star\leq 1\}\label{nuc subdif}
\eeq
In a very similar manner to the $\ell_1$ norm, subgradients of the nuclear norm depends only on the singular vectors of $\X_0$ and not on its singular values. Consequently, by normalizing singular values of $\X_0$ without modifying $\Sb$ we can make $f(\frac{\x_0}{\|\x_0\|_2})$ as high as $\sqrt{r}$. On the other hand, $\sup_{\Sb\in\pa \|\X_0\|_\star}\|\Sb\|_F$ is $\sqrt{d}$ and it is achieved when all singular values of $\Sb$ are equal to $1$. This yields,
\beq
\frac{R_{\x_0}}{f_{\max}(\x_0)}=\sqrt{\frac{d}{r}}\label{quant 2}
\eeq

\item{$\ell_{1,2}$ norm:} Let $\x_0\in\R^n$ for $n=tb$ as described in \eqref{eq:clo3}. Subdifferential of $\ell_{1,2}$ yields $R_{\x_0}=\sqrt{t}$, \cite{Decomp1}. Similar to $\ell_1$ norm, normalizing $\ell_2$ norms of the individual blocks of $\x_0$, it can be shown that $f_{\max}(\x_0)=\sqrt{k}$. Hence,
\beq
\frac{R_{\x_0}}{f_{\max}(\x_0)}=\sqrt{\frac{t}{k}}\label{quant 3}
\eeq
\end{itemize}
The reason behind these examples is to provide the reader with an intuition about the quantity $\frac{R_{\x_0}}{f_{\max}(\x_0)}$. From these examples, we observe the following relation,
\beq
\frac{R_{\x_0}}{f_{\max}(\x_0)}\approx\left(\frac{\text{Ambient dimension}}{\text{Degrees of freedom}}\right)^{1/2}
\eeq
For $\ell_1$ minimization, this is fairly clear as a sparse signal has $k$ degrees of freedom. For a $d\times d$ rank $r$ matrix, degrees of freedom is $r(2d-r)$, \cite{RechtNuc}, which lies between $dr$ and $2dr$ and the ambient dimension that the matrix belongs is $d^2$. From \eqref{quant 2}, we indeed have: $\sqrt{\frac{d}{r}}=\sqrt{\frac{d^2}{dr}}$.

\subsubsection{Simple bounds on optimal tuning for common functions}
From \eqref{quant 1} and \eqref{quant 2}, we can state the following results as straightforward observations.
\begin{propo}[Sandwiching optimally tuned NMSE] Let $\x_0\in\R^n$.
\begin{itemize}
\item Assume, $\x_0$ is a $k$ sparse signal (recall \eqref{eq:clo1}). Then, using ${\bf{D}}(\text{cone}(\pa\|\x_0\|_1))\geq k$, we have:
\beq
1\leq \frac{\min_{\la\geq 0}{\bf{D}}(\la\pa\|\x_0\|_1)}{{\bf{D}}(\text{cone}(\pa\|\x_0\|_1))}\leq 1+2\sqrt{\frac{n}{k^3}}
\eeq
\item Assume $\x_0$ corresponds to a rank $r$ matrix with size $d\times d$ where $n=d^2$ (recall \eqref{eq:clo2}). Using ${\bf{D}}(\text{cone}(\pa\|\x_0\|_\star))\geq rd$,
\beq
1\leq\frac{\min_{\la\geq 0}{\bf{D}}(\la\pa\|\x_0\|_\st)}{{\bf{D}}(\text{cone}(\pa\|\x_0\|_\st))}\leq 1+2\sqrt{\frac{1}{r^3d}}
\eeq
\item Assume $\x_0$ is a block sparse signals (recall (\ref{eq:clo3})) where $n=tb$ has $t$ blocks of size $b$ and $k$ of the blocks are nonzero. Using ${\bf{D}}(\text{cone}(\pa\|\x_0\|_{1,2}))\geq kb$,
\beq
1\leq\frac{\min_{\la\geq 0}{\bf{D}}(\la\pa\|\x_0\|_{1,2})}{{\bf{D}}(\text{cone}(\pa\|\x_0\|_{1,2}))}\leq 1+2\sqrt{\frac{t}{b^2k^3}}
\eeq
\end{itemize}
\end{propo}
\begin{proof}
For the proof, we combine \eqref{quant 1}, \eqref{quant 2}, \eqref{quant 3} with the facts that,
\beq
{\bf{D}}(\text{cone}(\pa\|\x_0\|_{1}))\geq k,~~~{\bf{D}}(\text{cone}(\pa\|\x_0\|_{\st}))\geq rd,~~~ {\bf{D}}(\text{cone}(\pa\|\x_0\|_{1,2}))\geq bk\label{super trivial}.
\eeq
Then, we apply \eqref{technical}. These bounds trivially arise from the specific structure of the sub-differentials of the $\ell_1$, nuclear norm and $\ell_{1,2}$ norm. For \eqref{super trivial}, the reader is referred to \cite{Decomp1,Negah2,simultaneous} for a detailed discussion. Note that $k$, $2dr-r^2$ and $bk$ are the ``degrees of freedom'' for the $k$ sparse signal, rank $r$ matrix and $k$ block-sparse signal hence \eqref{super trivial} is intuitive.
\end{proof}

\subsubsection{Alternative upper bounds}
We should emphasize that, different bounds can be developed. We will next consider a bound that does not involve the term $f(\x)$; however requires the following assumption.
\begin{assume}\label{assume 1} There exists a nontrivial subspace $T=T(f,\x_0)$ and a vector $\e=\e(f,\x_0)$ so that, for all $\s\in\paf$, we have:
\beq
\bu(\s,T)=\e
\eeq
\end{assume}
This assumption is related to the concept of decomposable norms, \cite{LPS,simultaneous,Negah2,Decomp1}, and is known to be true for sparsity inducing functions. If there are many such subspaces, we will choose the one with the largest dimension. In general, $T$ is called the \emph{support} of the sparse signal and $\e$ is known as the ``sign'' vector. For example, for $\ell_1$ norm, $T$ is the set of vectors, whose location of the nonzero entries are same as that of $\x_0$. On the other hand, $\e$ is simply the vector $\text{sgn}(\x_0)$.

For the nuclear norm described in \eqref{nuc subdif}, $T$ is the set of matrices $\X$ satisfying $(\Iden_{d}-\U\U^T)\Sb(\Iden_{d}-\V\V^T)=0$ and the sign vector is $\U\V^T$.

To be able to analyze $\min_{\tau\geq 0}\Dx\tau f)$, for a given vector $\vb$, we are interested in the optimal regularizer $\tau(\vb)$ that minimizes $\dt(\vb,\tau \pa f(\x_0))$ over all $\tau\geq 0$. In general, $\tau(\vb)$ may not be unique. In this case, we will use the smallest value. Assumption \ref{assume 1} implies an interesting property for $\tau(\vb)$, namely, $\tau(\vb)$ is a $\|\e\|_2^{-1}$-Lipschitz function of $\vb$. In other words, the optimal regularizer cannot change too much with a small change in $\vb$. More rigorously, for all $\vb_1,\vb_2\in\R^n$,
\beq
\|\tau(\vb_1)-\tau(\vb_2)\|_2\leq \frac{\|\vb_1-\vb_2\|_2}{\|\e\|_2}
\eeq

Based on this observation, we have the following result. The proof can be found in Lemma \ref{my result proof} of Appendix \ref{app swap}.
\begin{propo}\label{my bound} Suppose Assumption \ref{assume 1} holds. Let $L_{\x_0}=\|\e\|_2^{-1}$. Then,
\beq
\inf_{\tau\geq 0}\Dx\tau f)\leq \Delxf+2\pi(R_{\x_0}^2L_{\x_0}^2+R_{\x_0}L_{\x_0}\sqrt{\Delxf}+1)\label{up bound 2}
\eeq
\end{propo}
The right hand side of \eqref{up bound 2} involves the quantity $R_{\x_0}L_{\x_0}$. For the norms we considered, this quantity is actually equal to $\frac{R_{\x_0}}{f_{\max}(\x_0)}$ as we have $\|\e\|_2=f_{\max}(\x_0)$. Consequently, there is a close relation between the two upper bounds \eqref{technical} and \eqref{up bound 2}. Observe that, unlike \eqref{technical}, \eqref{up bound 2} only depends on the properties of the subdifferential set $\pa f(\x_0)$ rather than the particular value of $f(\x_0)$. 

Finally, we should emphasize that $\|\e\|_2^2$ is closely related to the ``generalized sparsity'' of the signal. The reader is referred to \cite{Negah2,Decomp1,LPS,simultaneous} for a comprehensive discussion.

\section{Calculation of the LASSO Cost}\label{sec proj error}
This section will make use of the results of \cite{McCoy}. Let us introduce the concept of statistical dimension and its relation to Gaussian squared-distance.
\begin{defn} [Statistical dimension, \cite{McCoy}] Let $\Cc\in\R^n$ be a closed and convex cone. Statistical dimension of $\Cc$ is denoted by $\delta(\Cc)$ and is defined as,
\beq
\delta(\Cc)=\E[\|\bu(\g,\Cc)\|_2^2],\nn
\eeq
where $\g\sim\Nn(0,\Iden_n)$.
\end{defn}
When $\Cc$ is a linear subspace, the statistical dimension reduces to the regular notion of dimension. In general, $\delta(\Cc)$ is a good measure of the \emph{size} of $\Cc$ and $\Cc$ behaves like a $\delta(\Cc)$ dimensional linear subspace, \cite{McCoy}. It is also closely related to the Gaussian width that has been the topics of papers \cite{Cha,ChaJor}. When $\Cc$ is a convex and closed cone, $\delta(\Cc)$ and $\DC$ are related as follows,
\beq
\delta(\Cc^*)=\DC\nn
\eeq
This follows from Moreau's decomposition (Fact \ref{more}) as follows,
\beq
\DC=\E[\dt(\g,\Cc)^2]=\E[\|\Pi(\g,\Cc)\|_2^2]=\E[\|\bu(\g,\Cc^*)\|_2^2]=\delta(\Cc^*)\nn
\eeq
Again using Moreau's decomposition, it can be shown that $\delta(\Cc)+\delta(\Cc^*)=\E[\|\g\|_2^2]=n$.

\subsection{Proof of Theorem \ref{LASSO}}
We first present the proof of Theorem \ref{LASSO}. To do this, we will make use of Theorem I of \cite{McCoy}; which characterizes the probability of nontrivial intersection for two randomly oriented cones. To rotate a cone $\Cc$ randomly, we will multiply its elements by a unitary matrix $\U$; which is drawn uniformly at random.
\beq
\U \Cc:=\{\U\ab\big|\ab\in\Cc\}\nn
\eeq
\begin{thm}[Kinematic Formula, \cite{McCoy}]\label{kinematics} Let $A,B$ be two arbitrary closed and convex cones, one of which is not a subspace. Let $\U$ be a unitary matrix chosen uniformly at random with respect to Haar measure. Then, for any $\eps>0$, the followings hold:
\begin{align}
&\delta(A)+\delta(B)\leq n-\eps \sqrt{n}\implies \Pro(A\cap \U B=\emptyset)\geq 1-4\exp(-\frac{\eps^2}{16})\nn\\
&\delta(A)+\delta(B)\geq n+\eps \sqrt{n}\implies \Pro(A\cap \U B=\emptyset)\leq4\exp(-\frac{\eps^2}{16})\nn
\end{align}
\end{thm}

Theorem \ref{kinematics} shows that, two cones will intersect with high probability when $\delta(A)+\delta(B)>n$. In order to prove Theorem \ref{LASSO}, we introduce a useful variation of Theorem \ref{kinematics}, which probabilistically characterizes the statistical dimension of the intersection when two cones do intersect. Basically, we show that, when $\delta(A)+\delta(B)$, with high probability, $\delta(A\cap \U B)$ is around $\delta(A)+\delta(B)-n$. The detailed discussion of this can be found in Appendix \ref{demixing appendix}. We now proceed with the proof of Theorem \ref{LASSO} and as a first step, we will prove the results on the projected LASSO error $\eta_{LASSO}$.

\subsection{Finding $\eta_{LASSO}$}\label{sec etaLASSO}
\begin{proof} [Proof of Theorem \ref{LASSO}, Calculation of $\eta_{LASSO}$] We will reduce \eqref{LASSOopt} to a problem that is equivalent to the regular constrained denoising \eqref{constrained prox}. Consider the set, 
\beq
\A\Cc=\{\A\x\big|\x\in\Cc\}\nn
\eeq
Hence, finding $\A\x^*$ is equivalent to finding the solution of:
\beq
\min_{\ub\in \A\Cc}\|\y-\ub\|_2\label{c proj}
\eeq
where $\y=\A\x_0+\sigma\vb$. From Theorem \ref{constrained}, we know that, for fixed $\A$ and Gaussian $\vb$, the worst case NMSE $\A(\x-\x_0)$ satisfies:
\beq
\eta_{LASSO}=\max_{\sigma> 0}\frac{\E[\|\A(\x-\x_0)\|_2^2]}{\sigma^2}={\bf{D}}(T_{\A\Cc}(\A\x_0)^*)\label{etalasso1}
\eeq
Hence, we simply need to characterize ${\bf{D}}(T_{\A\Cc}(\A\x_0)^*)$ to conclude. To do this, we will make use of the fact that $\A\in\R^{m\times n}$ is a uniformly random partial unitary matrix and we will first characterize $T_{\A\Cc}(\A\x_0)^*$. The following lemma, provides this characterization.
\begin{lem} \label{ACTangent}Assume $\A\A^T=\Iden_{m}$ and let $\Cc(\x_0,\A^T)=T_\Cc(\x_0)^*\cap \text{Range}(\A^T)$. Then,
\beq
T_{\A\Cc}(\A\x_0)^*=\A \Cc(\x_0,\A^T)\label{equal tangent}
\eeq
\end{lem}
\begin{proof} Assume $\s\in \A \Cc(\x_0,\A^T)$ and $\ub\in F_{\A\Cc}(\A\x_0)$. Then, there exists $\x\in\Cc$ such that, $\ub=\A(\x-\x_0)$ and $\s'\in \Cc(\x_0,\A^T)$ such that $\A\s'=\s$. Then,
\beq
\li\s,\ub\ri=\li\A\s',\A(\x-\x_0)\ri=\li\A^T\A\s',\A^T\A(\x-\x_0)\ri=\li\s',\A^T\A(\x-\x_0)\ri=\li\s',\x-\x_0\ri\leq 0\nn
\eeq
Hence $\s\in T_{\A\Cc}(\A\x_0)^*$. Conversely, assume $\s\not\in \A \Cc(\x_0,\A^T)$ and let $\s'=\A^T\s$. $\s'\not\in T_\Cc(\x_0)^*$ because $\s'\in\text{Range}(\A^T)$ however $\A\s'\not\in \A \Cc(\x_0,\A^T)$. Consequently, there exists $\x\in\Cc$ such that $\li\x-\x_0,\s'\ri>0$. It follows,
\beq
0<\li\x-\x_0,\s'\ri=\li\x-\x_0,\A^T\A\s'\ri=\li\A(\x-\x_0),\A\s'\ri=\li\A(\x-\x_0),\s\ri.\nn
\eeq
where $\A\x-\A\x_0\in T_{\A\Cc}(\A\x_0)$. This implies $\s\not\in T_{\A\Cc}(\A\x_0)^*$. Overall, we find \eqref{equal tangent}.
\end{proof}

Based on Lemma \ref{ACTangent}, we can write,
\beq
\eta_{LASSO}={\bf{D}}(T_{\A\Cc}(\A\x_0)^*)=m-\delta(T_{\A\Cc}(\A\x_0)^*)=m-\delta(\A \Cc(\x_0,\A^T)),\nn
\eeq
and calculate $\delta(\A \Cc(\x_0,\A^T))$ with high probability. Recall from Lemma \ref{ACTangent} that, $\Cc(\x_0,\A^T)$ is the intersection of a cone with the uniformly random subspace $\text{Range}(\A^T)$. Hence Theorem \ref{kinematics} is applicable. We will now split the problem in two cases.
\begin{itemize}
\item {\bf{Case 1 $m<\Delxc$:}} In this case, using Theorem \ref{kinematics} with $t=\frac{\Delxc-m}{\sqrt{n}}$, we find, with probability $1-c_1\exp(-c_2\frac{(m-\Delxc)^2}{n})$, $T_\Cc(\x_0)^*\cap \text{Range}(\A^T)=\{0\}$. Consequently, $T_{\A\Cc}(\A\x_0)^*=\A\{0\}=\{0\}$ which gives $\eta_{LASSO}=m$.
\item {\bf{Case 2 $m>\Delxc$:}} In this case, we use a modification of Theorem \ref{kinematics} which characterizes the statistical dimension of the intersection of a random subspace and a cone. This is, in fact, the topic of Appendix \ref{demixing appendix}. From Proposition \ref{inter dim}, there exists constants $c_1,c_2>0$ such that, with probability $1-c_1\exp(-c_2t^2)$, we have:
\beq
|\delta(T_\Cc(\x_0)^*\cap \text{Range}(\A^T))-(m-\Delxc)|\leq t\sqrt{n}\nn
\eeq
Equivalently, we have:
\beq
m-\Delxc-t\sqrt{n}\leq \delta(\Cc(\x_0,\A^T))\leq m-\Delxc+t\sqrt{n}\label{good ineq}
\eeq
On the other hand, since $\Cc(\x_0,\A^T)\subset\text{Range}(\A^T)$, multiplication with partial unitary $\A$ preserves the distances over $\text{Range}(\A^T)$ and is isometric, hence it will preserve the statistical dimension (see properties of statistical dimension in \cite{McCoy}) and will yield,
\beq
\delta(\A \Cc(\x_0,\A^T))=\delta( \Cc(\x_0,\A^T))\nn
\eeq
Overall, using \eqref{etalasso1} and applying Lemma \ref{ACTangent}, we have $\eta_{LASSO}=m-\delta( \Cc(\x_0,\A^T))$. Now, using \eqref{good ineq}, with the same probability, we find:
\beq
\Delxc-t\sqrt{n}\leq \eta_{LASSO}\leq \Delxc+t\sqrt{n}\nn
\eeq
\end{itemize}

\end{proof}

\subsection{Pinpointing the LASSO Cost}

The next result directly follows from characterization of $\eta_{LASSO}$ in Section \ref{sec etaLASSO} and Proposition \ref{prop objective} and gives the result on the LASSO cost $F_{LASSO}$. Basically, we use the fact that $F_{LASSO}+\eta_{LASSO}=m$.
\begin{propo} [LASSO cost] Assume $\A,\vb,\x_0,f$ is same as in Theorem \ref{LASSO}. Consider the cost function $f_{obj}(\x_0,\A,\vb)=\min_{\x\in\Cc}\|\y-\A\x\|_2^2$ in \eqref{LASSOopt}. Conditioned on $\A$, define the asymptotic LASSO cost as:
\beq
F_{LASSO}(\A)=\lim_{\sigma\rightarrow 0}\frac{\E[f_{obj}(\x_0,\A,\vb)]}{\sigma^2}\nn
\eeq
where the expectation is over $\vb\sim\Nn(0,\Iden_m)$. Then, there exists constants $c_1,c_2>0$ such that,
\begin{itemize} 
\item Whenever $m< \Delxc$, with probability $1-c_1\exp(-c_2\frac{(m-\Delxc)^2}{n})$, $F_{LASSO}=0$.
\item Whenever $m>\Delxc$, with probability $1-c_1\exp(-c_2t^2)$,
\beq
m-\Delxc-t\sqrt{n}\leq F_{LASSO}\leq m-\Delxc+t\sqrt{n}.\nn
\eeq
\end{itemize}
\end{propo}

%
%

\subsection{Numerical results}
{\bf{Remark:}} To relate Theorem \ref{LASSO} to the level-set constrained problem $\min_{f(\x)\leq f(\x_0)} \|\y-\A\x\|_2$, we will make use of the fact that $T_f(\x_0)^*=\text{cone}(\paf)$ when $\x_0$ is not a minimizer of $f(\cdot)$ (see \cite{Roc70}).

\subsubsection{LASSO with the $\ell_1$ norm}\label{sec ell1 LASSO}
We considered the following $\ell_1$ constrained optimization,
\beq\label{constrained ell1}
\min_{\x}\|\y-\A\x\|_2^2~~~\text{subject to}~~~\|\x\|_1\leq \|\x_0\|_1.
\eeq 
We let $\x_0$ to be a $k$ sparse vector and chose $k$ to be $20,40$ and $60$ while varying number of measurements $m$ from $20$ to $400$. The ambient dimension is $n=500$. We have performed $50$ realization of the problem in which $\A$ is generated as a random unitary matrix and $\sigma\vb$ is the noise vector, where $\vb\sim\Nn(0,\Iden)$ and $\sigma$ is sufficiently small.

We have estimated the quantities $\eta_{LASSO}$ and $F_{LASSO}$ by averaging $\|\y-\A\x^*\|_2^2$ and $\|\A\x^*-\A\x_0\|_2^2$ over $50$ realizations.

Finally, in order to verify our results, we estimate the term ${\bf{D}}(\text{cone}(\|\x_0\|_1))$. This is done by making use of the classical results on $\ell_1$ phase transitions, (see \cite{Don1,DonCentSym,Sto1}). In particular, see Theorem 4 of \cite{Sto1}. For example, when $n=500$ and $\x_0$ is $20$ sparse, we find ${\bf{D}}(\text{cone}(\|\x_0\|_1))\approx 89$. Similarly, $40$ and $60$ gives ${\bf{D}}(\text{cone}(\|\x_0\|_1))\approx 142$ and $186$ respectively.

Figure \ref{fig-Unitary} illustrates our results. In Figure \ref{fig-Error}, we observe that we can accurately predict projected LASSO error based on Theorem \ref{LASSO}. The dashed red line is what we theoretically expect and the blue, green and black markers are the experimental results for $k=20,40$ and $60$ respectively.

Figure \ref{fig-Obj} demonstrates that, LASSO cost can be similarly predicted. We should note that, in both figures there is an apparent phase transition. When the number of measurements are not sufficient, ($m<{\bf{D}}(\text{cone}(\|\x_0\|_1)))$, we observe that $\eta_{LASSO}$ increases linearly in $m$ actually it is equal to $m$. On the other hand, when the number of measurements are sufficient, $\eta_{LASSO}$ stays same and is simply equal to ${\bf{D}}(\text{cone}(\|\x_0\|_1))$. 

Recall that, $m>{\bf{D}}(\text{cone}(\|\x_0\|_1))$ regime is the regime where the noiseless compressed sensing problem succeeds. For $\ell_1$ minimization, \cite{NoiseSense} argues that, there is in fact a phase transition for noise sensitivity, and when $m>{\bf{D}}(\text{cone}(\|\x_0\|_1))$, the LASSO will recover $\x_0$ robustly and when $m<{\bf{D}}(\text{cone}(\|\x_0\|_1))$, the normalized error $\frac{\|\x^*-\x_0\|^2_2}{\sigma^2}$ will be unbounded as $\sigma\rightarrow 0$. Overall, we observe that the phase transition for $\eta_{LASSO}$ and $F_{LASSO}$ occurs exactly at $m\approx{\bf{D}}(\text{cone}(\|\x_0\|_1))$; which is also known to be the noise sensitivity threshold \cite{NoiseSense}.

It should be emphasized that the more popular formulation of LASSO is,
\beq
\min_\x\frac{1}{2}\|\y-\x\|_2^2+\la\|\x\|_1.\nn
\eeq
This has been studied in a series of papers \cite{NoiseSense,BayMon,Mon,AMPmain} and \cite{BicRit,BayMon,Mon} give analytical characterization of the asymptotic LASSO error $\E[\|\x^*-\x_0\|_2^2]$. On the other hand, our results hold for arbitrary convex functions and are not limited to $\ell_1$ minimization. Now, we will illustrate this with an example on nuclear norm.

\begin{figure}
\centering
\mbox{
\subfigure[]{
\includegraphics[width=3.1in]{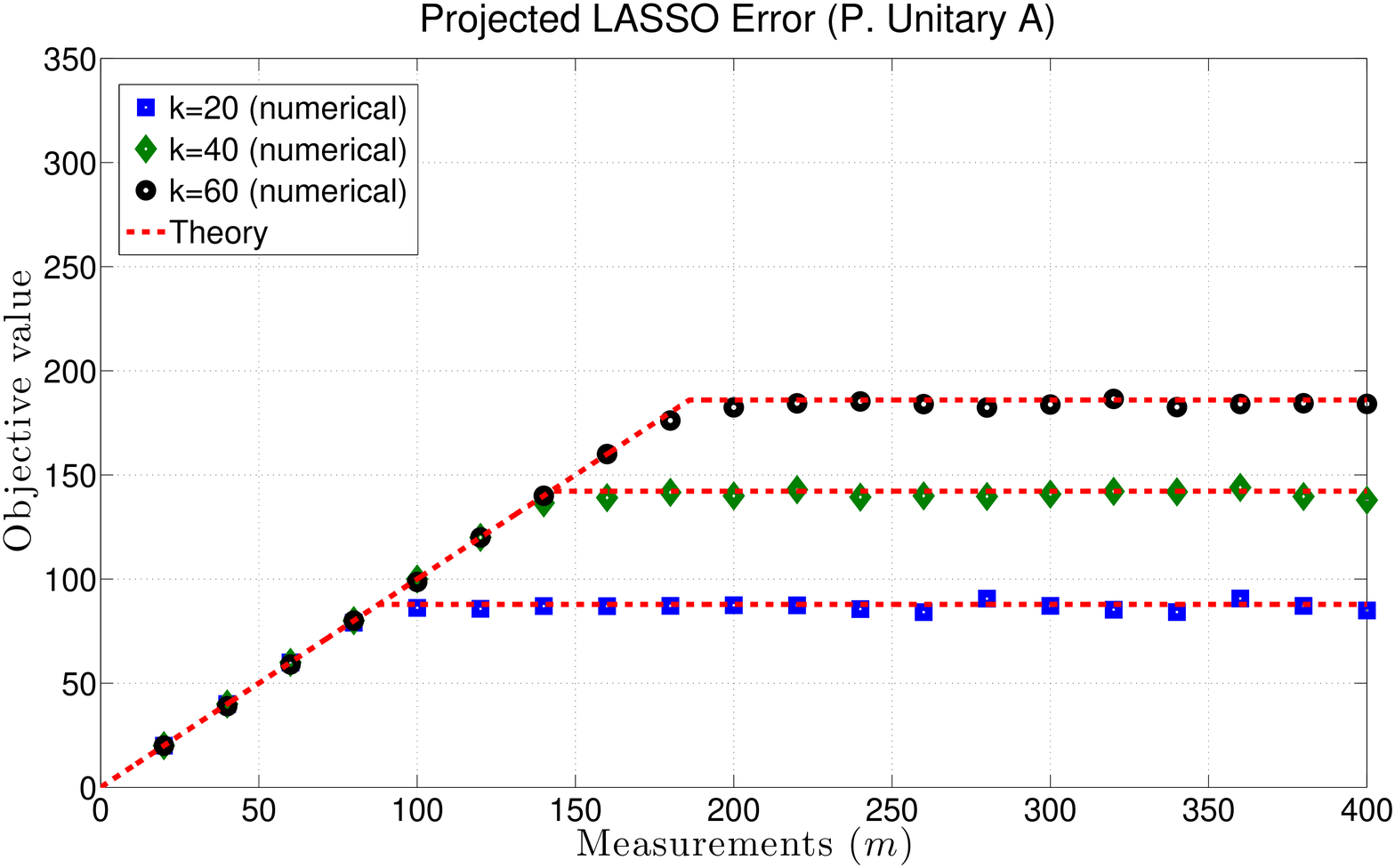}
\label{fig-Error}
}\quad
\subfigure[]{
\includegraphics[width=3.1in]{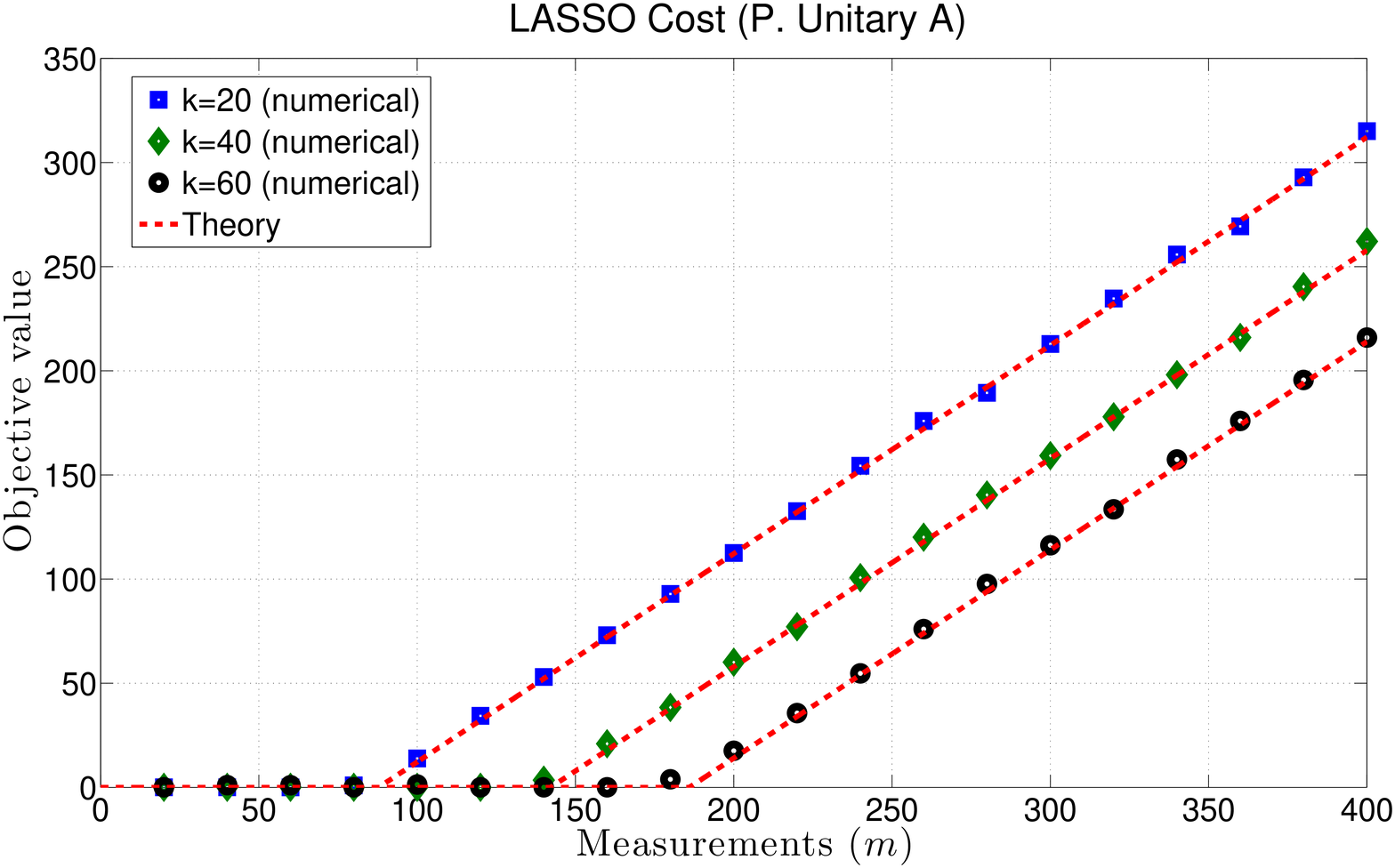}
\label{fig-Obj}
}
}
\centering
\caption{\small{$a)$ represents the expected projected LASSO error $\|\A(\x^*-\x_0)\|_2^2$ when $\A\in\R^{m\times n}$ is partial unitary. $b)$ represents the LASSO cost i.e. $\E[\|\y-\A\x^*\|_2^2]$. Results are for $\ell_1$ minimization, and for three different sparsity levels of $k=20,40$ and $60$. Vector length is $n=500$. Measurements run from $m=20$ to $400$.}}
\label{fig-Unitary}
\end{figure}

\subsubsection{LASSO with the nuclear norm}
We next considered a scenario where $\x_0$ is a vectorized form of a $30\times 30$, rank $4$ matrix. We observe $\y=\A\x_0+\sigma\vb$ where $\A$ is an $m\times 900$ partial unitary and $\vb\sim\Nn(0,\Iden)$. We solve,
\beq
\min_\x\|\y-\A\x\|_2^2~~~\text{subject to}~~~\|\x\|_\star\leq \|\x_0\|_\star\label{nuc lasso}
\eeq
where $\|\cdot\|_\star$ returns the nuclear norm of the $30\times 30$ matrix form of $\x_0$.

To predict the LASSO cost, we made use of the results of \cite{Oym}, which approximates ${\bf{D}}(\text{cone}(\|x_0\|_\star))$ based on the asymptotic singular value distribution of the $n\times n$ i.i.d.~Gaussian matrices in a similar manner to the $\ell_1$ phase transition calculations.

Using results of \cite{Oym}, we estimate ${\bf{D}}(\text{cone}(\|x_0\|_\star))\approx 389$. We then repeat the same experiment that is described in the previous section where we vary $m$ from $30$ to $900$ in the intervals of $30$ and average the results to estimate $\eta_{LASSO}$ and $F_{LASSO}$.

This time, we additionally approximated the actual asymptotic LASSO error $E_{LASSO}=\lim_{\sigma\rightarrow 0}\frac{\E[\|\x^*-\x_0\|_2^2]}{\sigma^2}$. Figure \ref{LASSOfigEvery} shows the projected error and the actual error as a function of $m$. We observe that, projected error can again be accurately predicted from our theoretical results. The LASSO is not robust for the regime $m<{\bf{D}}(\text{cone}(\|x_0\|_\star))$ and $E_{LASSO}$ is unbounded. Conversely, when $m>{\bf{D}}(\text{cone}(\|x_0\|_\star))$ the error becomes finite and decreases as a function of $m$. At $m=n=900$, we see that $E_{LASSO}=\eta_{LASSO}$. This is not surprising due to the fact that $\|\A(\x^*-\x_0)\|_2=\|\x^*-\x_0\|_2$ when $m=900$ as $\A$ becomes a full (square) unitary matrix and preserves the $\ell_2$ norm. Figure \ref{fig-DeltaLASSO} provides the LASSO cost for the same simulation. The cost satisfies $F_{LASSO}=m-\eta_{LASSO}$ as predicted by our theory. Figure \ref{fig-DeltaLASSO} also illustrates how $\Delxf$ corresponds to the stability phase transition of \eqref{nuc lasso} (recall Section \ref{PT discussion}).

\begin{figure}

  \begin{center}
{\includegraphics[scale=0.4]{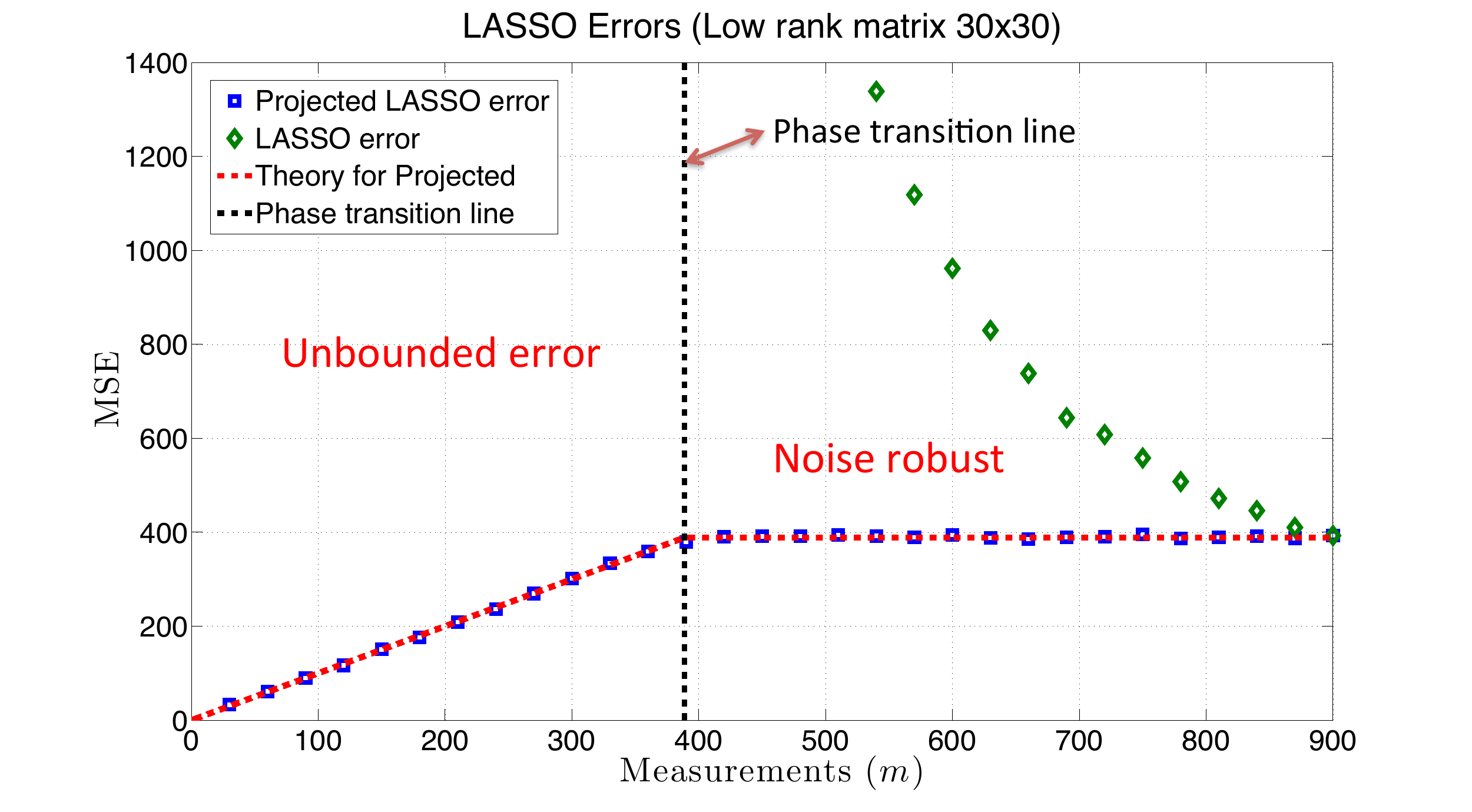}}  
  \end{center}
  \caption{\small{Simultation is performed for a $30\times 30$ matrix of rank $4$. The blue markers are the projected LASSO error; which matches with theory. Green markers are the (total) normalized LASSO MSE; which is unbounded on the left side of the phase transition line and a decreasing function of $m$ on the right side of the PT line $m=\Delxf$ (recall Section \ref{PT discussion}).}}
\label{LASSOfigEvery}
\end{figure}

\subsubsection{LASSO for Gaussian matrices}
We additionally performed the same simulations in Section \ref{sec ell1 LASSO} and solved \eqref{constrained ell1} for i.i.d. Gaussian $\A$ rather than partial unitary. Estimation of the quantities $\frac{\E[\|\A(\x^*-\x_0)\|_2^2]}{\sigma^2}$ and $\frac{\E[\|\y-\A\x^*\|_2^2]}{\sigma^2}$ are carried out in the in the exact same manner. In Figure \ref{fig-GaussObj} the dashed red line is the theoretical prediction for $\eta_{LASSO}$ when $\A$ is unitary and the blue, green and black markers are the results for i.i.d. Gaussian $\A$. Hence, Figure \ref{fig-GaussObj} indicates that, our results for partial unitary compression matrices also holds for Gaussian compression matrices and the results might be universal. This would not be surprising given the recent advances of the universality of the phase transitions, \cite{BLM12,Universal}.


\begin{figure}
\centering
\mbox{
\subfigure[]{
\includegraphics[width=3.1in]{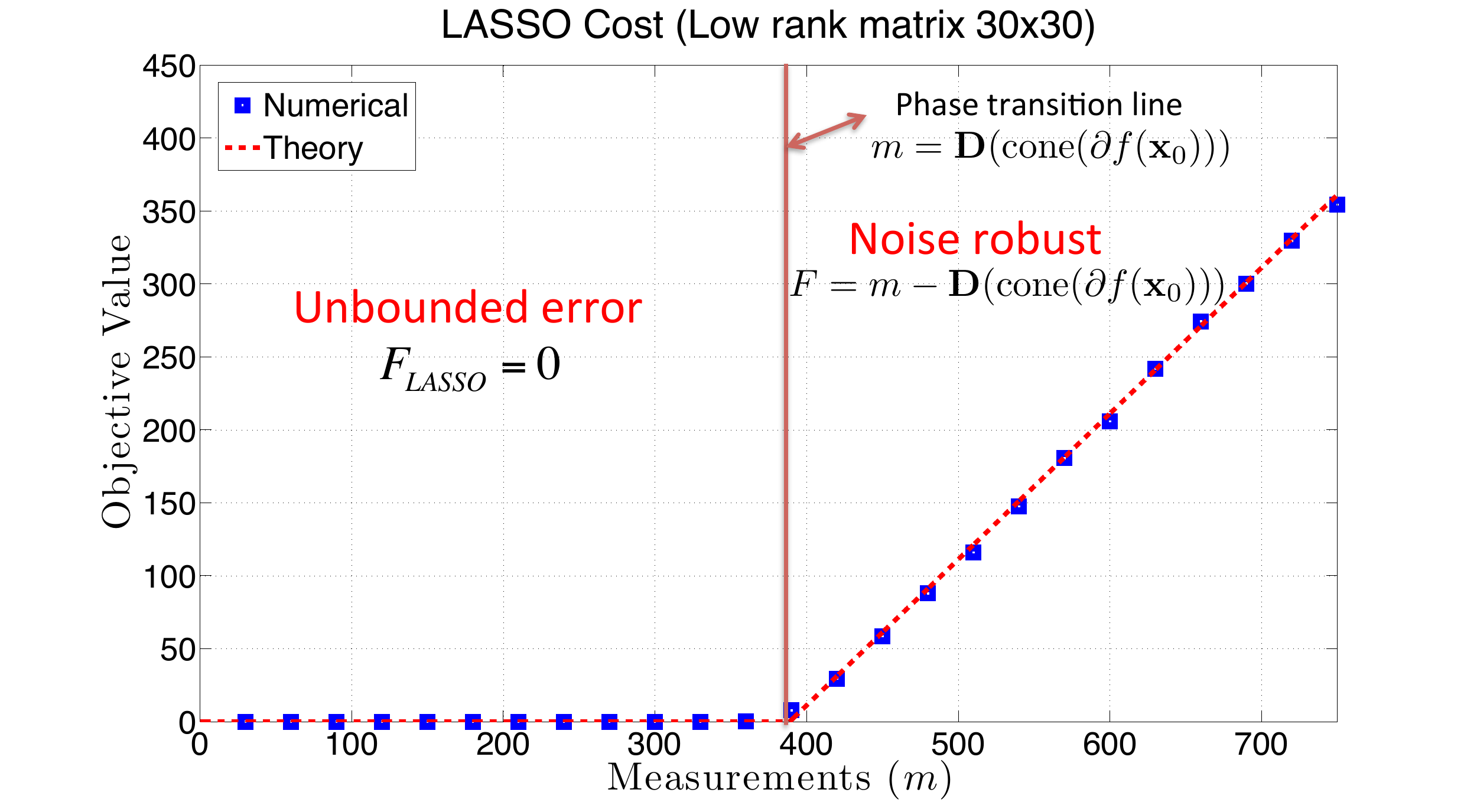}
\label{fig-DeltaLASSO}
}\quad
\subfigure[]{
\includegraphics[width=3.1in]{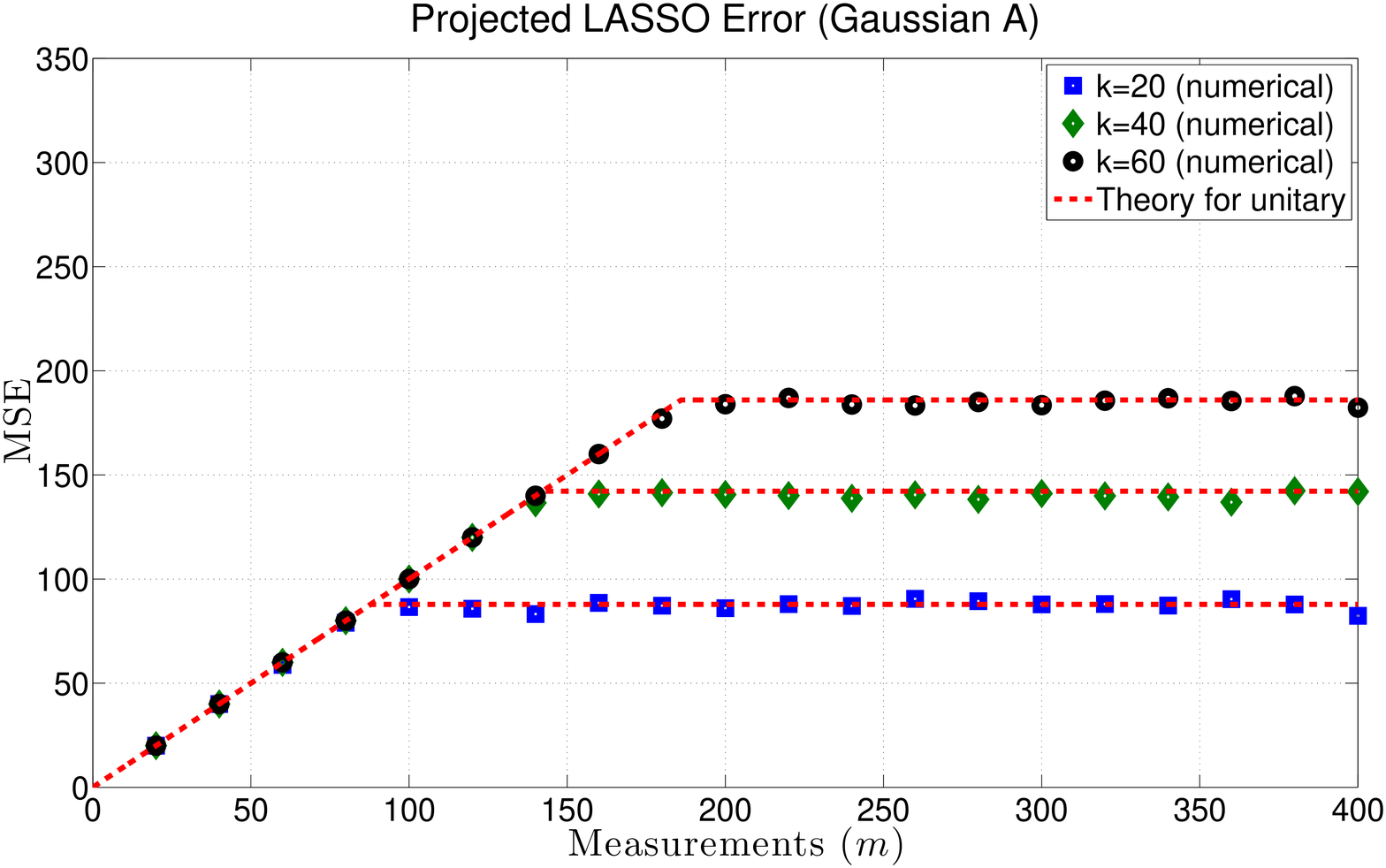}
\label{fig-GaussObj}
}
}
\centering
\caption{\small{a) The LASSO cost as a function of measurements for $30\times 30$, rank $4$ matrix. b) This figure is same as Figure \ref{fig-Error} except i.i.d.~Gaussian measurement is used rather than partial unitary.}}
\label{fig-Gaussian}
\end{figure}

\section{Discussion of the results}
We have considered the proximal denoising problem and provided sharp MSE upper bounds that are achievable in the small noise regime. Our bounds depend on the convex geometry of the problem \eqref{proxmain} and can be captured by distance to the scaled subdifferential $\la\paf$ or to the subdifferential cone $\text{cone}(\paf)$. These are meaningful quantities when $\x_0$ is a structured signal and $f(\cdot)$ is properly chosen to induce structure. Surprisingly, our estimation bounds are closely related to the recovery phase transition of the linear inverse problem \eqref{LinInv}.

We also showed an interesting phase transition for the generalized LASSO problem. When the number of measurements $m$ are smaller than $\Delxf$, there is no noise robustness and the cost value is equal to $0$. On the other hand, when the number of measurements are more than $\Delxf$, the cost is around $m-\Delxf$. This indicates that behavior of LASSO is closely related to the same quantity $\Delxf$.

We should again emphasize that, while $\ell_1$-minimization is often the primary interest, our results apply to all convex functions.


\subsection{Future directions}
\begin{itemize}
\item{Analysis of Generalized LASSO:} While we considered some basic properties of the generalized LASSO problem \eqref{LASSOopt}, the most critical one is yet to be explored. We believe there is a simple formula that predicts the LASSO error in,
\beq
\min_{\x}\|\y-\A\x\|_2^2+\la f(\x)
\eeq
which depends on the number of measurements $m$, the subdifferential $\paf$ and the regularizer $\la$. The generalized LASSO analysis will be a unification of the results on noiseless compressed sensing \eqref{LinInv} and compression-less denoising problem \eqref{proxmain}. Results of Chandrasekaran et al. and Amelunxen et al. \cite{Cha,McCoy} apply to noiseless linear inverse problem \eqref{LinInv} and we have mostly focused on estimation via convex functions without linear measurements. With generalized LASSO analysis, one will hopefully be able to predict the behavior of the noisy linear inverse problem and extend the results of \cite{Mon,BayMon,BicRit}, from $\ell_1$-minimization to arbitrary functions.

\item{Universality in denoising:} Assume the noise vector $\z$ has independent and identically distributed (i.i.d.) entries but the distribution is not normal. In this case, unfortunately our MSE bounds are no longer accurate. A simple example is the scenario where $\x_0$ is a sparse signal, and the entries of $\z$ are equally likely to be $\sigma$ and $-\sigma$. If $\x_0$ is $k$ sparse, ${\bf{D}}(\text{cone}(\|\x_0\|_1))\sim k\log\frac{2n}{k}$, \cite{DonCentSym}.

On the other hand, assuming $\sigma$ is sufficiently small and setting $\la=1$ in \eqref{proxmain2}, soft-thresholding reveals that the estimate $\x^*$ satisfies,
\beq
x^*_i=\begin{cases}0~\text{if}~x_{0,i}=0\\x_{0,i}-\text{sgn}(x_{0,i})\sigma+z_i~\text{else}\end{cases}
\eeq
Overall, this gives $\E[\frac{\|\x^*-\x_0\|_2^2}{\sigma^2}]=2k$; which is smaller than $\Delxf$ for the regime $k\ll n$.

While this shows that estimation error may depend on the distribution, one can actually introduce further randomization to the noise. Given $\z$, assume we observe $\x_0+\U\z$ where $\U$ is a uniformly random unitary matrix. We believe that, if one first generates and fixes $\U$ and then takes the expectation over $\z$, the worst case normalized MSE will in fact be around $\Delxf$, under mild assumptions. Possible such assumptions are i.i.d.'ness and subgaussianity of the entries of $\z$; which are used in \cite{BLM12} to show the universality of the compressed sensing phase transitions for $\ell_1$ minimization.

\end{itemize}

\section*{Acknowledgments}
Authors would like to thank Michael McCoy and Joel Tropp for stimulating discussions and helpful comments. Michael McCoy pointed out Lemma \ref{lem square lip} and informed us of various recent results most importantly Theorem \ref{kinematics}. S.O. would also like to thank his colleagues Kishore Jaganathan and Christos Thrampoulidis for their support.

\newpage

    \appendix
    \cleardoublepage
    \addcontentsline{toc}{section}{Appendix}
    \addtocontents{toc}{\protect\setcounter{tocdepth}{-1}}

\appendix
\vspace{7pt}
\begin{center}\Large{{APPENDIX}}\end{center}
\section{Auxiliary results}

\begin{fact}[Hyperplane separation theorem, \cite{Bertse}] \label{hypseper}Assume $\Cc_1,\Cc_2\subseteq\R^n$ are disjoint closed sets at least one of which is compact. Then, there exists a hyperplane $H$ such that $\Cc_1,\Cc_2$ lies on different open half planes induced by $H$.
\end{fact}

\begin{fact}[Properties of the projection, \cite{Bertse,Boyd}] \label{prom}Assume $\Cc\subseteq\R^n$ is a nonempty, closed and convex set and $\ab,\bb\in\R^n$ are arbitrary points. Then:
\beq
\|\bu(\ab)-\bu(\bb)\|_2\leq \|\ab-\bb\|_2\nn
\eeq
The projection $\bu(\ab,\Cc)$ is the unique vector satisfying,
\beq
\bu(\ab,\Cc)=\arg\min_{\vb\in\Cc}\|\ab-\vb\|_2\label{lem1}
\eeq
The projection $\bu(\ab,\Cc)$ is also the unique vector $\s_0$ that satisfies,
\beq
\li\s_0,\ab-\s_0\ri=\sup_{\s\in\Cc} \li\s,\ab-\s_0\ri\label{desiredlem2}
\eeq
In other words, $\ab$ and $\Cc$ lies on different half planes induced by the hyperplane that goes through $\bu(\ab,\Cc)$ and that is orthogonal to $\ab-\bu(\ab,\Cc)$.
\end{fact}
\begin{fact}[Moreau's decomposition theorem, \cite{More}] \label{more}Let $\Cc$ be a closed and convex cone in $\R^n$. For any $\vb\in\R^n$, the followings are equivalent:
\begin{itemize}
\item $\vb=\ab+\bb$, $\ab\in\Cc,\bb\in \Cc^*$ and $\ab^T\bb=0$.
\item $\ab=\bu(\vb,\Cc),\bb=\bu(\vb,\Cc^*)$.
\end{itemize}
\end{fact}

\begin{defn}[Lipschitz function]$h(\cdot):\R^n\rightarrow \R$ is called $L$-Lipschitz if for all $\x,\y\in\R^n$, $|h(\x)-h(\y)|\leq L\|\x-\y\|_2$.
\end{defn}
The next lemma provides a concentration inequality for Lipschitz functions of Gaussian vectors, \cite{Tal}.
\begin{fact}\label{classic} Let $\g\sim\Nn(0,I)$ and $h(\cdot):\R^n\rightarrow \R$ be an $L$-Lipschitz function. Then for all $t\geq 0$:
\beq
\Pro(|h(\g)-\E[h(\g)]|\geq t)\leq 2\exp(-\frac{t^2}{2L^2})\nn
\eeq
\end{fact}

\begin{lem}\label{ell2} For any $\g\sim\Nn(0,I)$, $c>1$, we have:
\beq
\Pro(\|\g\|_2\geq c\sqrt{n})\leq 2\exp(-\frac{(c-1)^2n}{2})\nn
\eeq
\end{lem}
\begin{proof}
$\E[\|\g\|_2]\leq \sqrt{\E[\|\g\|_2^2]}=\sqrt{n}$. Secondly $\ell_2$ norm is a $1$-Lipschitz function due to the triangle inequality. Hence:
\beq
\Pro(\|\g\|_2\geq c\sqrt{n})\leq \Pro(\|\g\|_2\geq (c-1)\sqrt{n}+\E[\|\g\|_2])\leq 2\exp(-\frac{(c-1)^2n}{2})\nn
\eeq
\end{proof}

\begin{lem}\label{more2}Let $\Cc$ be a closed and convex cone in $\R^n$. Then, ${\bf{D}}(\Cc)+{\bf{D}}(\Cc^*)=n$.
\end{lem}
\begin{proof} Using Fact \ref{more}, any $\vb\in\R^n$ can be written as $\bu(\vb,\Cc)+\bu(\vb,\Cc^*)=\vb$ and $\li\bu(\vb,\Cc),\bu(\vb,\Cc^*)\ri=0$. Hence,
\beq
\|\vb\|^2=\|\bu(\vb,\Cc^*)\|^2+\|\bu(\vb,\Cc)\|^2=\dt(\vb,\Cc)^2+\dt(\vb,\Cc^*)^2\nn
\eeq
Letting $\vb\sim\Nn(0,\Iden)$ and taking the expectations, we can conclude.
\end{proof}

\section{Subdifferential of the approximation}\label{approximate subdif}


\begin{proof}[Proof of Lemma \ref{approx subdif}]
Recall that $\hat{f}_{\x_0}(\x_0+\vb)-f(\x_0)$ is equal to the directional derivative $\fp\vb)=\sup_{\s\in\paf} \li\s,\vb\ri$. Also recall the ``set of maximizing subgradients'' from \eqref{dir dev2}. Clearly, $\pa\fp\vb)=\pa \hat{f}_{\x_0}(\x_0+\vb)$. We will let $\x=\w+\x_0$ and investigate $\pa\fp\w)$ as a function of $\w$.

\noindent{\bf{If $\w=0:$}} For any $\s\in\paf$ and any $\vb$ by definition, we have:
\beq
\fp\vb)-\fp0)=\fp\vb)=\sup_{\s'\in\paf} \li\vb,\s'\ri\geq \li\vb,\s\ri\nn
\eeq
hence $\s\in\pa \fp0)$. Conversely, assume $\s\not\in\paf$, then there exists $\vb$ such that:
\beq
f(\vb+\x_0)<f(\x_0)+\li\vb,\s\ri\nn
\eeq
By convexity for any $\eps>0$:
\beq
\frac{f(\eps\vb+\x_0)-f(\x_0)}{\eps}\leq f(\vb+\x_0)-f(\x_0)<\li\vb,\s\ri\nn
\eeq
Taking $\eps\rightarrow 0$ on the left hand side we find:
\beq
\fp\vb)-\fp0)=\fp\vb)<\li\vb,\s\ri\nn
\eeq
which implies $\s\not\in \pa \fp0)$.\\
{\bf{If $\w\neq 0$:}} Now, consider the case $\w\neq 0$. Assume, $\s\in \pafb \w)$. Then, for any $\vb$, we have:
\begin{align}
\fp\w+\vb)-\fp\w)&=\sup_{\s_1\in \paf} \li\w+\vb,\s_1\ri-\sup_{\s_2\in \paf} \li\w,\s_2\ri\\
&=\sup_{\s_1\in \paf} \li\w+\vb,\s_1\ri- \li\w,\s\ri\geq \li\vb,\s\ri
\end{align}
Hence, $\s\in\paw$. Conversely, assume $\s\not\in \pafb \w)$. Then, we'll argue that $\s\not\in \paw$.\\
Assume $\fp\w)=c\|\w\|_2^2$ for some scalar $c=c(\w)$. We can write $\s=a\w+\ub$ where $\ub^T\w=0$. Choose $\vb=\eps \w$ with $|\eps|<1$. We end up with:
\beq
\fp\w+\vb)-\fp\w)=\eps\sup_{\s_1\in \paf} \li\w,\s_1\ri=c\eps\|\w\|_2^2\geq \li\s,\vb\ri=a\eps\|\w\|_2^2\nn
\eeq
Consequently, we have $c\eps\geq a\eps$ for all $|\eps|<1$ which implies $a=c$. Hence, $\s$ can be written as $c\w+\ub$. Now, if $\s\in\paf$ then $\s\in \pafb \x-\x_0)$ as it maximizes $\li\s',\w\ri$ over $\s'\in\paf$. However we assumed $\s\not\in \pafb \x-\x_0)$. Observe that $\ub=\s-c\w$ and $\pafb \x-\x_0)-c\w$ lies on $n-1$ dimensional subspace $H$ that is perpendicular to $\w$. By assumption $\ub\not\in \pafb \x-\x_0)-c\w$. We'll argue that this leads to a contradiction. By making use of convexity of $\pafb \x-\x_0)-c\w$ and invoking Hyperplane separation theorem (Fact \ref{hypseper}), we can find a direction $\h\in H$ such that:
\beq
\li\h,\ub\ri>\sup_{\s'\in \pafb \x-\x_0)-c\w}\li\h,\s'\ri\label{good direction}
\eeq
Next, considering $\eps\h$ perturbation, we have:
\bea
\fp\w+\eps\h)-\fp\w)&=\sup_{\s_1\in \paf}( \eps\li\h,\s_1\ri-\sup_{\s_2\in \paf} \li\w,\s_2-\s_1\ri)
\end{align}
Denote the $\s_1$ that establish equality by $\s_1^*$.\\
{\bf{Claim:}} As $\eps\rightarrow 0$, $\li\s_1^*,\w\ri\rightarrow c\|\w\|_2^2$.
\begin{proof} Recall that $\paf$ is bounded. Let $R=\sup_{\s'\in\paf}\|\s'\|_2$. Choosing $\s_1\in \pafb \x-\x_0)$, we always have:
\beq
\fp\w+\eps\h)-\fp\w)\geq \eps\li\s_1,\h\ri\geq -\eps R\|\h\|_2\nn
\eeq
On the other hand, for any $\s_1$ we may write:
\beq
 \eps\li\h,\s_1\ri-\sup_{\s_2\in \paf} \li\w,\s_2-\s_1\ri\leq \eps R\|\h\|_2+\li\s_1,\w\ri-c\|\w\|_2^2\nn
\eeq
Hence, for $\s_1^*$, we obtain:
\beq
\eps R\|\h\|_2+\li\s_1^*,\w\ri-c\|\w\|_2^2\geq -\eps R\|\h\|_2\implies \li\s_1^*,\w\ri\geq c\|\w\|_2^2-2\eps R\|\h\|_2\nn
\eeq
Letting $\eps\rightarrow 0$, we obtain the desired result.
\end{proof}
\noindent{\bf{Claim:}} Given $\paf$, for any $\eps'>0$ there exists a $\delta>0$ such that for all $\s_1\in\paf$ satisfying $\li\s_1,\w\ri>c\|\w\|_2^2-\delta$ we have $\dt(\s_1,\pafb \x-\x_0))<\eps'$.
\begin{proof}
Assume for some $\eps'>0$, claim is false. Then, we can construct a sequence $\s(i)$ such that $\dt(\s(i),\pafb \x-\x_0))\geq\eps'$ but $\li\s(i),\w\ri\rightarrow c\|\w\|_2^2$. From the well-known Bolzano-Weierstrass Theorem and the compactness of $\paf\subseteq\R^n$, $\s(i)$ will have a convergent subsequence whose limit $\s(\infty)$ will be inside $\paf$ and will satisfy $\li\s_\infty,\w\ri=c\|\w\|_2^2=\fp\w)$. On the other hand, $\dt(\s(\infty),\pafb \x-\x_0))\geq \eps'\implies \s(\infty)\not\in \pafb \x-\x_0)$ which is a contradiction.
\end{proof}
Going back to what we have, using the first claim, as $\eps\h\rightarrow 0$, $\li\s_1^*,\w\ri\rightarrow c\|\w\|_2^2$. Using the second claim, this implies for some $\delta$ which approaches to $0$ as $\eps\rightarrow 0$, we have:
\beq
\sup_{\s_1\in \paf}( \eps\li\h,\s_1\ri-\sup_{\s_2\in \paf} \li\w,\s_2-\s_1\ri)\leq \eps(\delta\|\h\|_2+\sup_{\s'\in \pafb \x-\x_0)-c\w}\li\s',\h\ri)\nn
\eeq

Finally, based on \eqref{good direction}, whenever $\eps$ is chosen to ensure $\delta\|\h\|_2<\li\h,\ub\ri-\sup_{\s'\in \pafb \x-\x_0)-c\w}\li\s',\h\ri$ we have,
\beq
\fp\w+\eps\h)-\fp\w)<\eps\li\h,\ub\ri,\nn
\eeq
which contradicts with the initial assumption that $\s$ is a subgradient of $\fp\cdot)$ at $\w$, since,
\beq
\fp\w+\eps\h)-\fp\w)\geq \li\s,\eps\h\ri=\eps\li\ub,\h\ri.\nn
\eeq
\end{proof}


\begin{lem} \label{approx conv}$\hat{f}_{\x_0}(\x)$ is a convex function of $\x$.
\end{lem}
\begin{proof} To show convexity, we need to argue that the function $\fp\w)$ is a convex function of $\w=\x-\x_0$.\\
Observe that $g(\w)=f(\x_0+\w)-f(\x_0)$ is a convex function of $\w$ and behaves same as the directional derivative $\fp\w)$ for sufficiently small $\w$. More rigorously, from \eqref{dirder3}, for any $\w_1,\w_2\in\R^n$ and $ \delta>0$ there exists $\eps>0$ such that, we have:
\beq
g(\eps\w_1)\leq \fp\eps\w_1)+\delta\eps, ~g(\eps\w_2)\leq \fp\eps\w_2)+\delta\eps\nn
\eeq
Hence, for any $0\leq c\leq 1$:
\bea
\fp\eps(c\w_1+(1-c)\w_2))&\leq g(\eps(c\w_1+(1-c)\w_2))\nn\\
&\leq cg(\eps\w_1)+(1-c)g(\eps\w_2)\nn\\
&\leq c\fp\eps\w_1)+(1-c)\fp\eps\w_2)+\eps \delta\nn
\end{align}
Making use of the fact that $\fp\eps\s)=\eps \fp\s)$ for any direction $\s$, we obtain:
\beq
\fp c\w_1+(1-c)\w_2)\leq c\fp\w_1)+(1-c)\fp\w_2)+ \delta\nn
\eeq
Letting $\delta\rightarrow 0$, we may conclude with the convexity of $\fp\cdot)$ and problem (\ref{first or}).
\end{proof}


\section{Swapping the minimization over $\tau$ and the expectation}\label{app swap}
\begin{lem}[ \cite{Bod98,Led01}]\label{lem square lip}Assume $\g\sim\Nn(0,\Iden_n)$ and let $h(\cdot):\R^n\rightarrow \R$ be an $L$-Lipschitz function. Then, we have:
\beq
\text{Var}(h(\g))\leq L^2\nn
\eeq
\end{lem}
We next show a closely related result.
\begin{lem} \label{lem abs lip}Assume $\g\sim\Nn(0,\Iden_n)$ and let $h(\cdot):\R^n\rightarrow \R$ be an $L$-Lipschitz function. Then, we have:
\beq
|h(\g)-\E[h(\g)]|\leq \sqrt{2\pi}L\nn
\eeq
\end{lem}

\begin{proof} From Lipschitzness of $h(\cdot)$, letting $\ab=h(\g)-\E[h(\g)]$ and invoking Lemma \ref{classic} for all $t\geq 0$, we have:
\beq
\Pro(|\ab-\E[\ab]|\geq t)= \Pro(|\ab|\geq t)\leq 2\exp(-\frac{t^2}{2L^2})\nn
\eeq

Denote the probability density function of $|\ab|$ by $p(\cdot)$ and let $Q(u)=\Pro(|\ab|\geq u)$. We may write:
\beq
\E[|\ab|]=\int_{0}^{\infty}u p(u)du=\int_{\infty}^{0}udQ(u)=[uQ(u)]_{\infty}^0+\int_{0}^{\infty}Q(u)du\nn
\eeq
Using $Q(u)\leq 2\exp(-\frac{u^2}{2L^2})$ for $u\geq 0$, we have:
\beq
[uQ(u)]_{\infty}^0=[2u\exp(-\frac{u^2}{2L^2})]_{\infty}^0=0\nn
\eeq
Next,
\beq
\int_{0}^{\infty}Q(u)du\leq \int_{0}^{\infty}2\exp(-\frac{u^2}{2L^2})du=\sqrt{2\pi}L\nn
\eeq

\end{proof}


\begin{lem}\label{lemF3} Suppose Assumption \ref{assume 1} holds. Recall that $\tau(\vb)=\arg\min_{\tau\geq 0}\dt(\vb,\tau\paf)$. Then, for all $\vb_1,\vb_2$,
\beq
|\tau(\vb_1)-\tau(\vb_2)|\leq \frac{\|\vb_1-\vb_2\|_2}{\|\e\|_2}\label{eqlemF3}
\eeq
Hence, $\tau(\vb)$ is $\|\e\|_2^{-1}$-Lipschitz function of $\vb$.
\end{lem}
\begin{proof} Let $\ab_i=\bu(\vb_i,\cn(\Cc))$ for $1\leq i\leq 2$. Using Lemma \ref{prom}, we have $\|\ab_1-\ab_2\|_2\leq \|\vb_1-\vb_2\|_2$ as $\paf$ is convex. Now, we will further lower bound $\|\vb_1-\vb_2\|_2$ as follows:
\beq
\|\bu(\ab_1-\ab_2,T)\|_2\leq \|\ab_1-\ab_2\|_2\leq \|\vb_1-\vb_2\|_2\nn
\eeq
Now, observe that $\|\Pc_T(\ab_1-\ab_2)\|_2=\|\tau(\vb_1)\e-\tau(\vb_2)\e\|_2$. Hence, we may conclude with \eqref{eqlemF3}.
\end{proof}

\begin{lem}\label{my result proof} Let $\Cc$ be a convex and closed set. Define the set of $\tau$ that minimizes $\dt(\vb,\tau\Cc)$,
\beq
{\bf{T}}(\vb)=\{\tau\geq 0\big|\arg\min_{\tau\geq 0}\dt(\vb,\tau\Cc)\}\nn
\eeq
and let $\tau(\vb)=\inf_{\tau\in{\bf{T}}(\vb)}\tau$. $\tau(\vb)$ is uniquely determined, given $\Cc$ and $\vb$. Further, assume $\tau(\vb)$ is an $L$ Lipschitz function of $\vb$ and let $R:=R(\Cc)=\max_{\ub\in\Cc}\|\ub\|_2$. Then,
\beq
\min_{\tau\geq 0}\E[\dt(\g,\tau\Cc)^2]\leq \DCC+2\pi(R^2L^2+RL\sqrt{\DCC}+1)\nn
\eeq
\end{lem}
\begin{proof}

Let $\g\sim\Nn(0,\Iden)$ and let $\tau^*=\E[\tau(\g)]$. Now, from triangle inequality:
\beq
|\tau(\vb)-\tau^*|\leq t\implies \dt(\vb,\tau^* \Cc)\leq \dt(\vb,\tau(\vb) \Cc)+Rt\nn
\eeq
Consequently,
\beq
 \E[\dt(\g,\tau(\g)\Cc)]\leq \min_{\tau\geq 0}\E[\dt(\g,\tau\Cc)]\leq\E[\dt(\g,\tau^*\Cc)]\leq \E[\dt(\g,\tau(\g)\Cc)+R|\tau(\g)-\tau^*|]\nn
\eeq
This gives:
\beq
\E[\dt(\g,\tau^*\Cc)]-\E[\dt(\g,\tau(\g)\Cc)]\leq R\E[|\tau(\g)-\tau^*|]\nn
\eeq
Observing $\E[\dt(\g,\tau(\g)\Cc)]=\E[\dt(\g,\cn(\Cc))]\leq \sqrt{\DCC}$, and using Lemma \ref{lem abs lip} we find:
\beq
\E[\dt(\g,\tau^*\Cc)]-\sqrt{\DCC}\leq \sqrt{2\pi}RL\nn
\eeq
This yields:
\beq
\E[\dt(\g,\tau^*\Cc)]^2-\DCC\leq \sqrt{2\pi}RL(2\sqrt{\DCC}+\sqrt{2\pi}RL)\nn
\eeq
Using Lemma \ref{lem square lip} we have $\E[\dt(\g,\tau(\g)\Cc)]^2\geq \E[\dt(\g,\tau(\g)\Cc)^2]-1$, which gives:
\beq
\min_{\tau\geq 0}\E[\dt(\g,\tau\Cc)^2]-\DCC\leq 2\pi R^2L^2+2\sqrt{2\pi}RL\sqrt{\DCC}+1\nn
\eeq

\end{proof}
\section{Intersection of a cone and a subspace}\label{demixing appendix}

\subsection{Intersections of randomly oriented cones}

Based on Kinematic formula (Theorem \ref{kinematics}), one may find the following result on the intersection of the two cones. We first consider the scenario in which one of the cones is a subspace.

\begin{propo} [Intersection with a subspace]\label{inter dim}Let $A$ be a closed and convex cone and let $B$ be a linear subspace. Denote $\delta(A)+\delta(B)-n$ by $\delta(A,B)$. Assume the unitary $\U$ is generated uniformly at random. Given $\eps>0$, we have the following:
\begin{itemize}
\item If $\delta(A)+\delta(B)+\eps\sqrt{n}>n$,
\beq
\Pro(\delta(A\cap \U B)\geq \delta(A,B)+\eps\sqrt{n})\leq 8\exp(-\frac{\eps^2}{64})\nn
\eeq
\item $\Pro(\delta(A\cap \U B)\leq \delta(A,B)-\eps\sqrt{n})\leq 8\exp(-\frac{\eps^2}{64})$.
\end{itemize}
\end{propo}
\begin{proof} Denote $A\cap \U B$ by $C$. Let $H$ be a subspace with dimension $n-d$ chosen uniformly at random independent of $\U$. Observe that, $\U B\cap H$ is a $\delta(B)-d$ dimensional random subspace for $d<\delta(B)$. Hence, using Theorem \ref{kinematics} with $A$ and $\U B\cap H$ yields:
\begin{align}
&\delta(A)+\delta(B)-d\leq n-t \sqrt{n}\implies \Pro(A\cap \U B\cap H=\{0\})\geq 1-4\exp(-\frac{t^2}{16})\label{kininter1}\\
&\delta(A)+\delta(B)-d\geq n+t \sqrt{n}\implies \Pro(A\cap \U B\cap H=\{0\})\leq 4\exp(-\frac{t^2}{16})\label{kininter12}
\end{align}
Observe that \eqref{kininter1} is true even when $d\geq\delta(B)$ since if $d\geq \delta(B)$, $\U B\cap H=\{0\}$ with probability $1$.

\noindent{\bf{Proving the first statement:}} Let $\gamma=\delta(A)+\delta(B)-n$, $\gamma_{\eps}=\gamma+\eps\sqrt{n}$ and $\gamma_{\eps/2}=\gamma+\frac{\eps}{2}\sqrt{n}$. We assume $\gamma_{\eps}>0$. Observing $A\cap \U B\cap H=C\cap H$, we may write:
\begin{align}
& \Pro(C\cap H=\{0\})\leq  \Pro(C\cap H=\{0\}\big|\delta(C)\geq \gamma_{\eps})+\Pro(\delta(C)\leq \gamma_{\eps})\\
&\text{and}~~~\Pro(\delta(C)\leq \gamma_{\eps})\geq \Pro(C\cap H=\{0\})-  \Pro(C\cap H=\{0\}\big|\delta(C)\geq \gamma_{\eps})\label{kininter2}
\end{align}
If, $\gamma_{\eps}>n$, $\Pro(\delta(C)\leq \gamma_{\eps})=1$. Otherwise, choose $d=\max\{\gamma_{\eps/2},0\}$.

\noindent{\bf{Case 1:}} If $d=0$, then, $\gamma_{\eps/2}\leq 0$ and $H=\R^n$. This gives,
\beq
\Pro(C\cap H=\{0\}\big|\delta(C)\geq \gamma_{\eps})=\Pro(C=\{0\}\big|\delta(C)\geq \gamma_{\eps})=0\label{kininter13}
\eeq
Also, choosing $t=\frac{\eps}{2}\sqrt{n}$ in \eqref{kininter1} and using $\gamma\leq -\frac{\eps}{2}\sqrt{n}$, we obtain:
\beq
\Pro(C\cap H=\{0\})=\Pro(C=\{0\})\geq 1-4\exp(-\frac{\eps^2}{64})\label{kininter14}
\eeq
\noindent{\bf{Case 2:}} Otherwise, $d=\gamma_{\eps/2}>0$. Applying Theorem \ref{kinematics}, we find:
\beq
\Pro(C\cap H=\{0\}\big|\delta(C)\geq \gamma_{\eps})\leq 4\exp(-\frac{\eps^2}{64})\label{kininter3}
\eeq
Next, choosing $t=\frac{\eps}{2}\sqrt{n}$ in \eqref{kininter1}, we obtain:
\beq
\Pro(C\cap H=\{0\})\geq 1-4\exp(-\frac{\eps^2}{64})\label{kininter4}
\eeq
Overall, combining \eqref{kininter2}, \eqref{kininter13}, \eqref{kininter14}, \eqref{kininter3} and \eqref{kininter4}, we obtain:
\beq
\Pro(\delta(C)\leq \gamma_{\eps})\geq 1-8\exp(-\frac{\eps^2}{64})\nn
\eeq
\noindent{\bf{Proving the second statement:}} 
In the exact same manner, this time, let $\gamma_{-\eps}=\gamma-\eps\sqrt{n}$, $\gamma_{-\eps/2}=\gamma-\frac{\eps}{2}\sqrt{n}$. If $\gamma_{-\eps}<0$,
\beq
\Pro(\delta(C)\leq \gamma_{-\eps})\leq \Pro(\delta(C)<0)=0\nn
\eeq
Otherwise, let $d=\gamma_{-\eps/2}$, we may write,
\beq
\Pro(\delta(C)\geq \gamma_{-\eps})\geq \Pro(C\cap H\neq\{0\})-  \Pro(C\cap H\neq\{0\}\big|\delta(C)\leq \gamma_{-\eps})\label{ABCC}
\eeq
in an identical way to \eqref{kininter2}. Repeating the previous argument and using \eqref{kininter12}, we may first obtain,
\beq
\Pro(C\cap H\neq\{0\})\geq 1-4\exp(-\frac{\eps^2}{64})\nn
\eeq
and using Theorem \ref{kinematics},
\beq
 \Pro(C\cap H\neq\{0\}\big|\delta(C)\leq \gamma_{-\eps})\leq 4\exp(-\frac{\eps^2}{64})\nn
\eeq
Combining these, gives the desired result.
\beq
\Pro(\delta(C)\geq \gamma_{-\eps})\geq 1-8\exp(-\frac{\eps^2}{64})\nn
\eeq
\end{proof}

\section{Proof of Theorem \ref{constrained} -- Lower bound}\label{constrained case}

\begin{thm} \label{finish constrained} Let $\Cc$ be a closed and convex set, $\vb\sim\Nn(0,\Iden)$ and let $\x^*(\sigma\vb)=\arg\min_{\x\in\Cc} \|\x_0+\sigma\vb-\x\|_2$. Then, we have,
\beq
\lim_{\sigma\rightarrow 0}\frac{\E[\|\x^*(\sigma\vb)-\x_0\|_2^2]}{\sigma^2}=\Delxc\nn
\eeq
\end{thm}
\begin{proof} Let $1\geq \alpha,\eps> 0$ be numbers to be determined. Denote probability density function of a $\Nn(0,c\Iden)$ distributed vector by $p_c(\cdot)$. From Lemma \ref{simple lem}, the expected error $\E[\|\x^*-\x_0\|_2^2]$ is simply,
\beq
\int_{\vb\in\R^n}\|\bu(\sigma\vb,F_\Cc(\x_0))\|_2^2p_1(\vb)d\vb\nn
\eeq
Let $S_{\alpha}$ be the set satisfying:
\beq
S_{\alpha}=\{\ub\in\R^n|\frac{\|\bu(\ub,T_\Cc(\x_0))\|_2}{\|\ub\|_2}\geq \alpha\}.\nn
\eeq
Let $\bar{S}_\alpha=\R^n-S_\alpha$. Using Proposition \ref{cone cool}, given $\eps>0$, choose $\eps_0>0$ such that, for all $\|\ub\|_2\leq \eps_0$ and $\ub\in S_{\alpha}$, we have,
\beq\|\bu(\ub,F_\Cc(\x_0))\|_2\geq (1-\eps)\|\bu(\ub,T_\Cc(\x_0))\|_2.\label{initial assumption}\eeq

Now, let $\z=\sigma\vb$. Split the error into three groups, namely:
\begin{itemize}
\item $F_1=\int_{\|\z\|_2\leq \eps_0,\z\in S_{\alpha}}\|\bu(\z,F_\Cc(\x_0))\|_2^2p_\sigma(\z)d\z$,~~ $T_1=\int_{\|\vb\|_2\leq \frac{\eps_0}{\sigma},\vb\in S_{\alpha}}\|\bu(\vb,T_\Cc(\x_0))\|_2^2p_1(\vb)d\vb$.
\item $F_2=\int_{\|\z\|_2\geq \eps_0,\z\in S_{\alpha}}\|\bu(\z,F_\Cc(\x_0))\|_2^2p_\sigma(\z)d\z$, ~~$T_2=\int_{\|\vb\|_2\geq \frac{\eps_0}{\sigma},\vb\in S_{\alpha}}\|\bu(\vb,T_\Cc(\x_0))\|_2^2p_1(\vb)d\vb$.
\item $F_3=\int_{\z\in \bar{S}_{\alpha}}\|\bu(\z,F_\Cc(\x_0))\|_2^2p_\sigma(\z)d\z$, ~~~~~~~~~~~$T_3=\int_{\vb\in \bar{S}_{\alpha}}\|\bu(\vb,T_\Cc(\x_0))\|_2^2p_1(\g)d\vb$.
\end{itemize}
The rest of the argument will be very similar to the proof of Proposition \ref{asymp}. We know the following from Proposition \ref{upper summary}:
\begin{align}
&T_1+T_2+T_3=\Delxc\nn\\
&F_1+F_2+F_3=\E[\|\x^*-\x_0\|_2^2]\leq \sigma^2(T_1+T_2+T_3)\nn
\end{align}
To proceed, we will argue that, the contributions of the second and third terms are small for sufficiently small $\sigma,\alpha,\eps>0$. Observe that:
\beq
T_3\leq \int_{\vb\in \bar{S}_{\alpha}}\alpha^2\|\vb\|_2^2p_1(\vb)d\vb\leq \alpha^2n\nn
\eeq
For $T_2$, we have:
\beq
T_2\leq \int_{\|\vb\|_2\geq \frac{\eps_0}{\sigma}}\|\vb\|_2^2p_1(\g)d\vb=C(\frac{\eps_0}{\sigma})\nn
\eeq
Since $\|\g\|_2$ has finite second moment, fixing $\eps_0>0$ and letting $\sigma\rightarrow 0$, we have $C(\frac{\eps_0}{\sigma})\rightarrow 0$.
For $T_1$, from \eqref{initial assumption}, we have:
\beq
F_1\geq (1-\eps)^2\sigma^2T_1\nn
\eeq
Overall, we found:
\beq
\frac{\E[\|\x^*-\x_0\|_2^2]}{\sigma^2}\geq \frac{F_1}{\sigma^2}\geq (1-\eps)^2\frac{T_1}{T_1+T_2+T_3}\Delxc\nn
\eeq
Writing $T_1=\Delxc-T_2-T_3\geq \Delxc-\alpha^2n-C(\frac{\eps_0}{\sigma})$, we have:
\beq
\frac{T_1}{T_1+T_2+T_3}\geq \frac{\Delxc-\alpha^2n-C(\frac{\eps_0}{\sigma})}{\Delxc}\nn
\eeq
Letting $\sigma\rightarrow 0$ for fixed $\alpha,\eps_0,\eps$ , we obtain:
\beq
\lim_{\sigma\rightarrow 0}\frac{\E[\|\x^*-\x_0\|_2^2]}{\sigma^2\Delxc}\geq (1-\eps)^2\frac{\Delxc-\alpha^2n}{\Delxc}\nn
\eeq
Since, $\alpha,\eps$ can be made arbitrarily small, we obtain $\lim_{\sigma\rightarrow 0}\frac{\E[\|\x^*-\x_0\|_2^2]}{\sigma^2\Delxc}=1$.
%

\end{proof}
The next result shows that, as $\sigma\rightarrow 0$, we can exactly predict the cost of the constrained problem.
\begin{propo} \label{prop objective} Consider the setup in Theorem \ref{finish constrained}. Let $\w^*(\sigma\vb)=\x^*(\sigma\vb)-\x_0$. Then,
\beq
\lim_{\sigma\rightarrow 0}\frac{\E[\|\sigma\vb-\w^*(\sigma\vb)\|_2^2]}{\sigma^2}={\bf{D}}(T_\Cc(\x_0))\nn
\eeq 
\end{propo}
\begin{proof} Let $\w^*=\w^*(\sigma\vb)$ and $\z=\sigma\vb$. $\z-\w^*$ satisfies two conditions.
\begin{itemize}
\item From Lemma \ref{simple lem}, $\|\z-\w^*\|_2=\dt(\z,F_\Cc(\x_0))\geq \dt(\z,T_\Cc(\x_0))$.
\item Using Lemma \ref{angle1}, $\|\z-\w^*\|_2^2+\|\w^*\|_2^2\leq \|\z\|_2^2$.
\end{itemize}
Consequently, when $\vb\sim\Nn(0,\Iden)$, we find:
\beq
n\sigma^2=\E[\|\z\|_2^2]\geq \E[\|\z-\w^*\|_2^2]+\E[\|\w^*\|_2^2]\geq \sigma^2\E[\bu(\vb,T_\Cc(\x_0)^*)^2]+\E[\|\w^*\|_2^2]\nn
\eeq
Normalizing both sides by $\sigma^2$ and subtracting ${\bf{D}}(T_\Cc(\x_0))=\E[\bu(\vb,T_\Cc(\x_0)^*)^2]$ and $\E[\|\w^*\|_2^2]$ we find:
\beq
\Delxc-\frac{\E[\|\w^*\|_2^2]}{\sigma^2}\geq \frac{\E[\|\z-\w^*\|_2^2]}{\sigma^2}-{\bf{D}}(T_\Cc(\x_0))\geq 0\nn
\eeq
where we used Lemma \ref{more2}. Now, letting $\sigma\rightarrow 0$ and using the fact that $\lim_{\sigma\rightarrow0}\frac{\E[\|\w^*\|_2^2]}{\sigma^2}=\Delxc$, we find the desired result.
\end{proof}

\section{Approximation results on convex cones}
%
%
\subsection{Standard observations}
{\bf{Remark:}} Throughout the section, $\Cc$ will be a nonempty, closed and convex set in $\R^n$. 
\begin{prop}\label{simple lem} Let $\x_0\in \Cc$ and $\y=\x_0+\z\in\R^n$. From Lemma \ref{lem1}, recall that, $\bu(\y,\Cc)$ is the unique vector that is equal to $\arg\min_{\ub\in\Cc}\|\y-\ub\|_2$. By definition of feasible set $F_\Cc(\x_0)$, we also have, $\bu(\y,\Cc)=\bu(\z,F_\Cc(\x_0))$.
\end{prop}

\begin{lem}\label{fees less tan} For all $\z\in\R^n$ and $\x_0\in\Cc$, we have:
\beq
\|\bu(\z,F_\Cc(\x_0))\|_2\leq \|\bu(\z,T_\Cc(\x_0))\|_2\nn
\eeq
\end{lem}
\begin{proof} Setting $f(\cdot)=0$ and $\y=\x_0+\z$ in Lemma \ref{upper bound lem}, we have:
\beq
\|\bu(\z,F_\Cc(\x_0))\|_2=\|\x^*-\x_0\|_2\leq \dt(\z,T_\Cc(\x_0)^*)=\|\bu(\z,T_\Cc(\x_0))\|_2\nn
\eeq
\end{proof}

The following lemma shows that projection onto the feasible cone is arbitrarily close to the projection onto the tangent cone as we scale down the vector. This is due to Proposition 5.3.5 of Chapter III of \cite{Urru}.
\begin{lem}\label{lemj3} Assume $\x_0\in\Cc$. Then, for any $\z\in\Cc$,
\beq
\lim_{\eps\rightarrow 0} \frac{\bu(\eps\w,F_\Cc(\x_0))}{\eps}\rightarrow \bu(\w,T_\Cc(\x_0))\nn
\eeq
Hence,
\begin{itemize}
\item If $\bu(\w,T_\Cc(\x_0))=0$, using Lemma \ref{fees less tan}, $\bu(\z,F_\Cc(\x_0))=0$.
\item If $\bu(\w,T_\Cc(\x_0))\neq 0$, 
\beq
\lim_{\eps\rightarrow 0}\frac{\|\bu(\eps\w,F_\Cc(\x_0))\|_2}{\|\bu(\eps\w,T_\Cc(\x_0))\|_2}=1.\nn
\eeq
\end{itemize}
\end{lem}
\subsection{Uniform approximation to the tangent cone}
\begin{propo}\label{cone cool} Let $\Cc$ be a closed and convex set including $\x_0$. Denote the unit $\ell_2$-sphere in $\R^n$ by $\Sc^{n-1}$ and let $1\geq \alpha>0$ be arbitrary. Given $\alpha,\eps>0$, there exists an  $\eps_0>0$ such that, for all $\w\in\Sc^{n-1}$, $\|\bu(\w,T_\Cc(\x_0))\|_2\geq \alpha$ and for all $0<t\leq \eps_0$, we have:
\beq
\frac{\|\bu(t\w,F_\Cc(\x_0))\|_2}{t\|\bu(\w,T_\Cc(\x_0))\|_2}\geq 1-\eps\label{main state}
\eeq
In particular, setting $\alpha=1$, given $\eps>0$, there exists $\eps_0>0$ such that, for all $t\leq \eps_0$ and all $\w\in T_\Cc(\x_0)\cap \Sc^{n-1}$, $\|\bu(t\w,F_\Cc(\x_0))\|_2\geq (1-\eps)t$.
\end{propo}
\noindent {\bf{Remark:}} Note that, statements of Propositions \ref{cone cool} and \ref{unif lem} are quite similar.
\begin{proof} Given $\alpha>0$, consider the following set:
\beq
S=\{\w\in\Sc^{n-1}|\|\bu(\w,T_\Cc(\x_0))\|_2\geq \alpha\}\nn
\eeq
This set is closed and bounded and hence compact. Define the following function on this set:
\beq
c(\w)=\max\{c>0 \big| \frac{\|\bu(c\w,F_\Cc(\x_0)\|_2}{\|\bu(c\w,T_\Cc(\x_0))\|_2}\geq 1-\eps\}\nn
\eeq
$c(\w)$ is strictly positive due to Lemma \ref{lemj3} and it can be as high as infinity. Furthermore, from Lemma \ref{continuous increase}, we know that whenever $c<c(\w)$:
\beq
\frac{\|\bu(c\w,F_\Cc(\x_0)\|_2}{\|\bu(c\w,T_\Cc(\x_0))\|_2}\geq 1-\eps\nn
\eeq
as well. Let $s(\w)=\min\{1,c(\w)\}$. If $s(\w)$ is continuous, since $\Sc^{n-1}$ is compact $s(\w)$ will attain its minimum which implies $c(\w)\geq s(\w)\geq \eps_0>0$ for some $\eps_0$. Again, this also implies, for all $\w\in \Sc^{n-1}$, and $0<t\leq \eps_0$, 
\beq
 \frac{\|\bu(t\w,F_\Cc(\x_0)\|_2}{\|t\bu(\w,T_\Cc(\x_0))\|_2}\geq 1-\eps\nn
\eeq

To end the proof, we will show continuity of $s(\w)$.

\noindent{\bf{Claim:}} $s(\w)$ is continuous.

\begin{proof} 
We will show that $\lim_{\w_2\rightarrow \w_1}s(\w_2)=s(\w_1)$. To do this, we will make use of the continuity of the functions $\|\bu(c_1\w,F_\Cc(\x_0)\|_2$, $\|\bu(c_1\w,T_\Cc(\x_0)\|_2$ and $\frac{\|\bu(c_1\w,F_\Cc(\x_0)\|_2}{\|\bu(c_1\w,T_\Cc(\x_0)\|_2}$ when the denominator is nonzero. Given $\w_1$, let $c_1=\min\{2,c(\w_1)\}$.

\noindent {\bf{Case 1:}} If $\frac{\|\bu(c_1\w_1,F_\Cc(\x_0)\|_2}{\|\bu(c_1\w_1,T_\Cc(\x_0))\|_2}>1-\eps$, then $c(\w_1)>2$ and for all $\w_2$ sufficiently close to $\w_1$, $\frac{\|\bu(c_1\w_2,F_\Cc(\x_0)\|_2}{\|\bu(c_1\w_2,T_\Cc(\x_0))\|_2}$ is more than $1-\eps$ and hence $c(\w_2)\geq2>1$. Hence, $s(\w_1)=s(\w_2)$.

\noindent {\bf{Case 2:}} Now, assume $\frac{\|\bu(c_1\w_1,F_\Cc(\x_0)\|_2}{\|\bu(c_1\w_1,T_\Cc(\x_0))\|_2}=1-\eps $ which implies $c_1=c(\w_1)$. Using the ``strict decrease'' part of Lemma \ref{continuous increase}, for any $\eps'>0$ and $c'=c_1-\eps'$, $\frac{\|\bu(c'\w_1,F_\Cc(\x_0)\|_2}{\|\bu(c'\w_1,T_\Cc(\x_0))\|_2}>1-\eps$. Then, for $\w_2$ sufficiently close to $\w_1$, $\frac{\|\bu(c'\w_2,F_\Cc(\x_0)\|_2}{\|\bu(c'\w_2,T_\Cc(\x_0))\|_2}>1-\eps$ which implies $c(\w_2)\geq c'$. Hence, $c(\w_2)\geq c_1-\eps'$ for arbitrarily small $\eps'>0$. Conversely, for any $\eps'>0$ and $c'=c_1+\eps'$, $\frac{\|\bu(c'\w_1,F_\Cc(\x_0)\|_2}{\|\bu(c'\w_1,T_\Cc(\x_0))\|_2}<1-\eps$. Then, for $\w_2$ sufficiently close to $\w_1$, $\frac{\|\bu(c'\w_2,F_\Cc(\x_0)\|_2}{\|\bu(c'\w_2,T_\Cc(\x_0))\|_2}<1-\eps$ which implies $c(\w_2)\leq c'$. Hence, $c(\w_2)\leq c_1+\eps'$ for arbitrarily small $\eps'>0$. Combining these, we obtain $c(\w_2)\rightarrow c(\w_1)$ as $\w_2\rightarrow \w_1$. This also implies $s(\w_2)\rightarrow s(\w_1)$.
\end{proof}
This finishes the proof of the main statement \eqref{main state}. For the $\alpha=1$ case, observe that, $\|\w\|_2=1$ and $\|\bu(\w,T_\Cc(\x_0))\|_2=1$ implies $\w\in T_\Cc(\x_0)$.
\end{proof} 

\begin{lem} \label{continuous increase} Let $\x_0\in\R^n$ and let $\w$ have unit $\ell_2$-norm and set $l_T=\|\bu(\w,T_\Cc(\x_0))\|_2$. Define the function,
\beq
g(t)=\begin{cases}\frac{\|\bu(t\w,F_\Cc(\x_0))\|_2}{t}~\text{for}~t>0\\l_T~\text{for}~t=0\end{cases}\nn
\eeq
Then, $g(\cdot)$ is continuous and non increasing on $[0,\infty)$. Furthermore, it is strictly decreasing on the interval $[t_0,\infty)$ where $t_0=\sup_{t} \{t>0\big|g(t)=l_T\}$.
\end{lem}
\begin{proof} Due to Lemma \ref{fees less tan}, $g(t)\leq l_T$ and from Lemma \ref{lemj3}, the function is continuous at $0$. Continuity at $t\neq 0$ follows from the continuity of the projection (see Fact \ref{prom}). Next, if $g(t)=l_T$,  using the fact that $F_\Cc(\x_0)$ contains $0$, the second statement of Lemma \ref{angle1} gives,
\beq
\bu(t\w,T_\Cc(\x_0))=\bu(t\w,F_\Cc(\x_0))\in F_\Cc(\x_0).\nn
\eeq
From convexity, $\bu(t'\w,T_\Cc(\x_0))\in F_\Cc(\x_0)$ for all $0\leq t'\leq t$. Hence, $g(t')=l_T$. This implies $g(t)=l_T$ for $t\leq t_0$.

Now, assume $t_1>t_0$ and $t_1>t_2>0$ for some $t_1,t_2>0$. Then, $g(t_1)<l_T$, hence, the third statement of Lemma \ref{angle1} applies. Setting $\alpha=\frac{t_2}{t_1}$ in Lemma \ref{angle1}, we find,
\beq
\|\bu(t_1\w,F_\Cc(\x_0))\|_2<\frac{\|\bu(t_2\w,F_\Cc(\x_0))\|_2}{\frac{t_2}{t_1}},\nn
\eeq
which implies the strict decrease of $\frac{\|\bu(t\w,F_\Cc(\x_0))\|_2}{t}$ over $t\geq t_0$.

%
%
%
\end{proof}

For the rest of the discussion, given three points $A,B,C$ in $\R^n$, the angle induced by the lines $AB$ and $BC$ will be denoted by $A\hat{B}C$.
\begin{lem} \label{angle1} Let $\Kc$ be a convex and closed set in $\R^n$ that includes $0$. Let $\z\in\R^n$ and $0<\alpha<1$ be arbitrary, let $\p_1=\bu(\z,\Kc)$, $\p_2=\bu(\alpha\z,\Kc)$. Denote the points whose coordinates are determined by $0,\p_1,\p_2,\z$ by $O, P_1,P_2$ and $Z$ respectively. Then,
\begin{itemize}
\item $Z\hat{P_1}O$ is either wide or right angle.
\item If $Z\hat{P_1}O$ is right angle, then $\p_1=\frac{\p_2}{\alpha}=\bu(\z,T_\Kc(0))$.
\item If $Z\hat{P_1}O$ is wide angle, then $\|\p_1\|_2<\frac{\|\p_2\|_2}{\alpha}\leq \|\bu(\z,T_\Kc(0))\|_2$.
\end{itemize}
\end{lem}
\begin{proof} {\bf{Acute angle:}} Assume $Z\hat{P_1}O$ is acute angle. If $Z\hat{O}P_1$ is right or wide angle, then $0$ is closer to $\z$ than $\p_1$ which is a contradiction. If $Z\hat{O}P_1$ is acute angle, then draw the perpendicular from $Z$ to the line $OP_1$. The intersection is in $\Kc$ due to convexity and it is closer to $\z$ than $\p_1$, which again is a contradiction.

 {\bf{Right angle:}} Now, assume $Z\hat{P_1}O$ is right angle. Using Fact \ref{prom}, there exists a hyperplane $H$ that separates $\z$ and $\Kc$ passing through $P_1$ which is perpendicular to $\z-\p_1$. The line $P_1O$ lies on $H$. Consequently, for any $\alpha\in [0,1]$, the closest point to $\alpha\z$ over $\Kc$ is simply $\alpha\p_1$. Hence, $\p_2=\alpha\p_1$. Now, let $\q_1:=\bu(\z,T_\Kc(0))$. Then, $\bu(\alpha\z,T_\Kc(0))=\alpha\q_1$. If $\q_1\neq \p_1$ then, $\|\q_1\|_2>\|\p_1\|_2$ since $\|\z-\q_1\|_2<\|\z-\p_1\|_2$ and:
\beq
\|\q_1\|_2^2=\|\z\|_2^2-\|\z-\q_1\|_2^2> \|\z\|_2^2-\|\z-\p_1\|_2^2\geq \|\p_1\|_2^2\nn
\eeq
where the last inequality follows from the fact that $Z\hat{P_1}O$ is not acute. Then,
\beq
\lim_{\alpha\rightarrow 0}\frac{\|\bu(\z,T_\Kc(0))\|_2}{\|\bu(\z,\Kc)\|_2}=\frac{\|\q_1\|_2}{\|\p_1\|_2}>1\nn
\eeq
which contradicts with Lemma \ref{lemj3}.

 {\bf{Wide angle:}} Finally, assume $Z\hat{P_1}O$ is wide angle. We start by reducing the problem to a two dimensional one. Obtain $\Kc'$ by projecting the set $\Kc$ to the $2D$ plane induced by the points $Z,P_1$ and $O$. Now, let $\p_2'=\bu(\alpha\z,\Kc')$. Due to the projection, we still have:
\beq
\|\z-\p_2'\|_2\leq\|\z-\p_2\|_2\leq \|\z-\alpha\p_1\|_2\label{good proj prop}
\eeq
and $\|\p_2'\|_2\leq\|\p_2\|_2$. Next, we will prove that $\|\p_2'\|_2>\|\alpha\p_1\|_2$ to conclude. Figure \ref{cases} will help us explain our approach. Let the line $UP_1$ be perpendicular to $ZP_1$. Assume, it crosses $ZO$ at $S$. Let $P'Z'$ be parallel to $P_1Z_1$. Observe that $P'$ corresponds to $\alpha \p_1$. $H$ is the intersection of $P'Z'$ and $P_1U$. Denote the point corresponding to $\p_2'$ by $P_2'$. Observe that $P_2'$ satisfies the following:
\begin{itemize}
\item $P_1$ is the closest point to $Z$ in $\Kc$ hence $P_2'$ lies on the left of $P_1U$ (same side as $O$).
\item $P_2$ is the closest point to $Z'$. Hence, $Z'\hat{P_2}P_1$ is not acute angle. Otherwise, we can draw a perpendicular to $P_2P_1$ from $Z'$ and end up with a shorted distance. This would also imply that $Z'\hat{P_2'}P_1$ is not acute as well. The reason is, due to projection, $|Z'P_2'|\leq |Z'P_2|$ and $|P_2'P_1|\leq |P_2P_1|$ hence,
\beq
|Z'P_1|\geq|Z'P_2|^2+|P_2P_1|^2\geq |Z'P_2'|^2+|P_2'P_1|^2\label{angle sum}
\eeq
\item $P_2'$ has to lie below or on the line $OP_1$ otherwise, perpendicular to $OP_1$ from $Z'$ would yield a shorter distance than $|P_2'Z'|$.

\item $\p_2\neq\alpha\p_1$. To see this, note that $Z'\hat{P'}O$ is wide angle. Let $\q\in\R^n$ be the projection of $\alpha\z$ on the line $\{c\p_1\big|c\in\R\}$ and point $Q$ denote the vector $\q$. If $Q$ lies between $O$ and $P_1$, $\q\in\Kc$ and $|QZ'|<|P'Z'|$. Otherwise, $P_1$ lies between $Q$ and $P'$ hence $|P_1Z'|<|P'Z'|$ and $\p\in\Kc$. This implies $P_2,P_2'\neq P'$.
\end{itemize}

Based on these observations, we investigate the problem in two cases illustrated by Figure \ref{cases}.
\begin{figure}

  \begin{center}
{\includegraphics[scale=0.25]{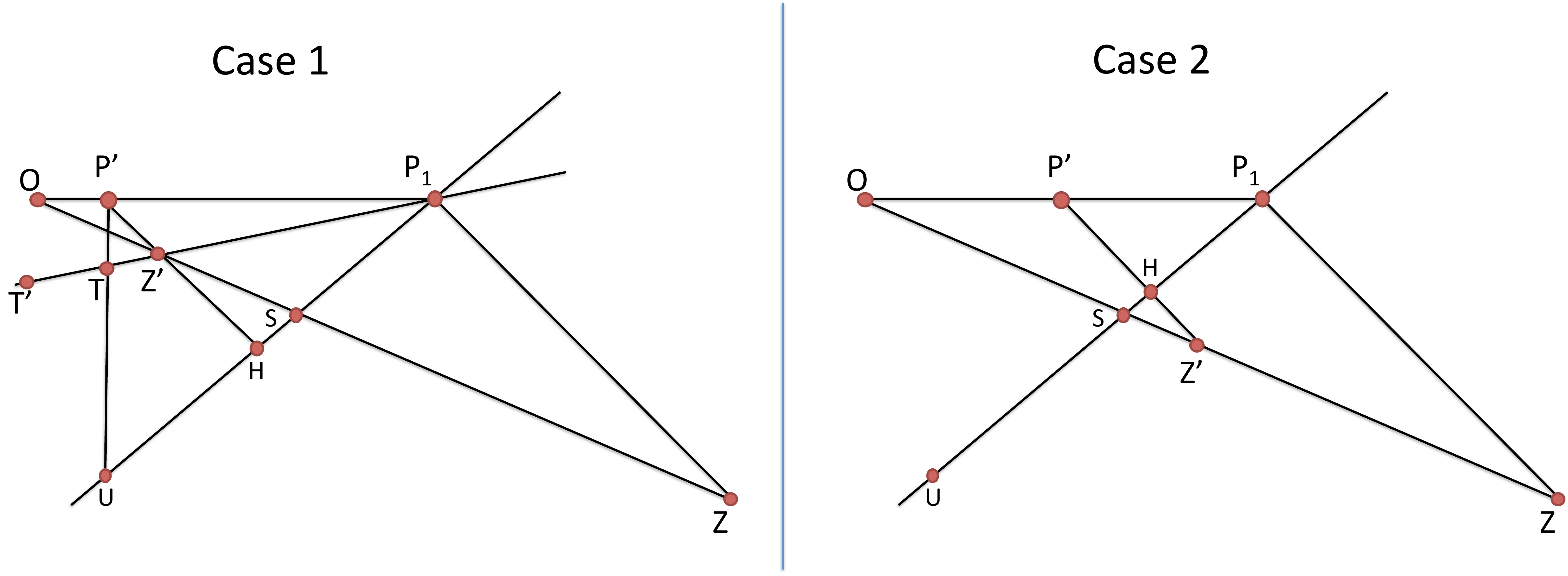}}  
  \end{center}
  \caption{Possible configurations of the points in Lemma \ref{angle1}.}
 \label{cases}
\end{figure}

\noindent {\bf{Case 1 ($S$ lies on $Z'Z$):}} Consider the lefthand side of Figure \ref{cases}. If $P_2'$ lies on the righthand side of $P'U$, this implies $|P_2'O|> |P'O|$ which is what we wanted.

If $P_2'$ lies on the region induced by $OP'TT'$ then $P_1\hat{P}_2'Z'$ is acute angle as $P_1\hat{Z}'P_2'>P_1\hat{Z}'P'$ is wide, which contradicts with \eqref{angle sum}.

If $P_2'$ lies on the remaining region $T'TU$, then $Z'\hat{P}_2'P_1$ is acute. The reason is, $P'_2\hat{Z}'P_1$ is wide as follows:
\beq
P'_2\hat{Z}'P_1\geq P'_2\hat{T}P_1\geq U\hat{T}P_1>U\hat{P}'P_1=\frac{\pi}{2}\nn
\eeq

\noindent {\bf{Case 2 ($S$ lies on $OZ'$):}} Consider the righthand side of Figure \ref{cases}. Due to location restrictions, $P_2'$ lies on either $P_1P'H$ triangle or the region induced by $OP'HU$. If it lies on $P_1P'H$ then, $O\hat{P'}P_2'\geq O\hat{P'}H$ (thus wide); which implies $|OP_2'|>|OP'|$ as $O\hat{P}'P_2'$ is wide angle and $P'\neq P_2'$.

If $P_2'$ lies on $OP'HU$ then, $P_1\hat{P}_2'Z'<P_1\hat{H}Z'=\frac{\pi}{2}$ hence $P_1\hat{P}'_2Z'$ is acute angle which contradicts with \eqref{angle sum}.

In all cases, we end up with $|OP_2'|>|OP'|$ which implies $\|\p_2\|_2\geq \|\p_2'\|_2>\alpha\|\p_1\|_2$ as desired.

Finally, apply Lemma \ref{fees less tan} on $\alpha\z$ to upper bound $\|\p_2\|_2$ by $\alpha\|\bu(\z,T_\Kc(0))\|_2$.
\end{proof}

\end{document}